%% file: coulomb.tex
\documentclass[11pt]{article}

\usepackage{titlesec}
\titleformat{\subsection}[runin]{\normalfont\bfseries}{\thesubsection.}{.5em}{}[.]\titlespacing{\subsection}{0pt}{2ex plus .1ex minus .2ex}{.8em}
\titleformat{\subsubsection}[runin]{\normalfont\itshape}{\thesubsubsection.}{.3em}{}[.]\titlespacing{\subsubsection}{0pt}{1ex plus .1ex minus .2ex}{.5em}
\titleformat{\paragraph}[runin]{\normalfont\itshape}{\theparagraph.}{.3em}{}[.]\titlespacing{\paragraph}{0pt}{1ex plus .1ex minus .2ex}{.5em}

\usepackage[labelfont=sc,font=small,labelsep=period]{caption}
\usepackage[letterpaper, hmargin=1.1in, top=1in, bottom=1.2in, footskip=0.6in]{geometry}

\makeatletter\def\SetFigFont#1#2#3#4#5{\small}
\usepackage{amsmath} %[intlimits]
\usepackage{amssymb}
\usepackage{amsfonts}
\usepackage{latexsym}
\usepackage{amsthm}
\usepackage{amsxtra}
\usepackage{amscd}
\usepackage{mathrsfs}
\usepackage{bm}

\usepackage{graphicx, color}
\usepackage[nottoc,notlof,notlot]{tocbibind} 
\usepackage{cite}

\numberwithin{equation}{section}
\numberwithin{figure}{section}
\theoremstyle{plain} %plain, definition, remark
\newtheorem{theorem}{Theorem}[section]
\newtheorem*{theorem*}{Theorem}
\newtheorem{lemma}[theorem]{Lemma}
\newtheorem*{lemma*}{Lemma}

\newtheorem*{corollary*}{Corollary}
\newtheorem{proposition}[theorem]{Proposition}
\newtheorem*{proposition*}{Proposition}
\newtheorem{definition}[theorem]{Definition}
\newtheorem*{definition*}{Definition}

\newtheorem*{conjecture*}{Conjecture}

\newtheorem*{example*}{Example}
\newtheorem{remark}[theorem]{Remark}
\newtheorem*{remark*}{Remark}

\renewcommand{\b}[1]{\boldsymbol{\mathrm{#1}}} %bold
\renewcommand{\r}{\mathrm}  %upright
\renewcommand{\cal}{\mathcal}

\definecolor{darkred}{rgb}{0.9,0,0.3}
\definecolor{darkblue}{rgb}{0,0.3,0.9}

\newcommand{\cob}{\color{darkblue}}

\usepackage{ifthen}
\def\comment#1{\ifthenelse{\isodd{\value{page}}}{\marginpar{\raggedright\scriptsize{\textcolor{darkred}{#1}}}}{\marginpar{\raggedleft\scriptsize{\textcolor{darkred}{#1}}}}}

\renewcommand{\P}{\mathbb{P}}
\newcommand{\E}{\mathbb{E}}
\newcommand{\R}{\mathbb{R}}
\newcommand{\C}{\mathbb{C}}
\newcommand{\D}{\mathbb{D}}

\newcommand{\Z}{\mathbb{Z}}

\newcommand{\ee}{\mathrm{e}}
\newcommand{\ii}{\mathrm{i}}

\renewcommand{\leq}{\leqslant}
\renewcommand{\geq}{\geqslant}
\renewcommand{\epsilon}{\varepsilon}

\newcommand{\ind}[1]{\b 1_{#1}}
\newcommand{\indb}[1]{\b 1_{\hb{#1}}}
\newcommand{\inda}[1]{\b 1_{\ha{#1}}}

\newcommand{\p}[1]{({#1})}

\newcommand{\pB}[1]{\Bigl({#1}\Bigr)}
\newcommand{\pbb}[1]{\biggl({#1}\biggr)}

\newcommand{\pa}[1]{\left({#1}\right)}

\newcommand{\hb}[1]{\bigl\{{#1}\bigr\}}

\newcommand{\ha}[1]{\left\{{#1}\right\}}

\newcommand{\absa}[1]{\left\lvert #1 \right\rvert}

\newcommand{\normB}[1]{\Bigl\lVert #1 \Bigr\rVert}

\DeclareMathOperator{\supp}{supp}

\DeclareMathOperator{\dist}{dist}
\DeclareMathOperator{\diam}{diam}

\newcommand{\ddp}[2]{\frac{\partial #1}{\partial #2}}
\newcommand{\half}{\frac{1}{2}}
\newcommand{\hatmueps}{\hat\mu^{(\varepsilon)}}

\begin{document}

\title{Local density for two-dimensional one-component plasma}
\author{Roland Bauerschmidt\footnote{Harvard University, Department of Mathematics. E-mail: {\tt brt@math.harvard.edu}.} \and
Paul Bourgade\footnote{New York University, Courant Institute of Mathematical Sciences. E-mail: {\tt bourgade@cims.nyu.edu}.} \and
Miika Nikula\footnote{Harvard University, Center of Mathematical Sciences and Applications. E-mail: {\tt minikula@cmsa.fas.harvard.edu}.} \and
Horng-Tzer Yau\footnote{Harvard University, Department of Mathematics. E-mail: {\tt htyau@math.harvard.edu}.}}
\date{}
\maketitle

\begin{abstract}
  We study the classical two-dimensional one-component plasma
  of $N$ positively charged point particles,
  interacting via the Coulomb potential and confined by an external potential.
  For the specific inverse temperature $\beta=1$ (in our normalization),
  the charges are the eigenvalues of random normal matrices,
  and the model is exactly solvable as a determinantal point process.
  For any positive temperature, 
  using a multiscale scheme of iterated mean-field bounds,
  we prove that the equilibrium measure provides the local particle density down to the optimal scale of $N^{o(1)}$ particles.
  Using this result and the loop equation, we further prove that the
  particle configurations are rigid,
  in the sense that the fluctuations of smooth linear statistics on any scale are $N^{o(1)}$.
\end{abstract}

\section{Introduction and results}

\subsection{One-component plasma}

Given a potential $V: \C \to \R \cup \{+\infty\}$, 
the energy of a configuration of $N$ charges $z=(z_1,\dots, z_N) \in \C^N$ is defined by
\begin{equation} \label{e:Hdef}
  H_{N,V}(z) = \sum_{j \neq k} \log \frac{1}{|z_j-z_k|} + N \sum_{j} V(z_j).
\end{equation}
The two-dimensional one-component plasma (OCP) of $N$ charges at inverse temperature $\beta>0$
is the Gibbs measure on $\C^N$ defined by
\begin{equation} \label{e:Pdef}
  P_{N,V,\beta}(dz) = \frac{1}{Z_{N,V,\beta}} \ee^{-\beta H_{N,V}(z)} \, m^{\otimes N}(dz)
  ,
\end{equation}
where $m$ denotes the Lebesgue measure on $\C$ and $Z_{N,V,\beta}$ a normalization constant
(assuming that $V$ has sufficient growth at infinity, so that the latter is well-defined).
For notational convenience, we use $\beta$ rather than $\beta/2$ in \eqref{e:Pdef}.
In particular, in our normalization, the exactly solvable case is $\beta=1$ rather than $\beta=2$,
differently from the usual normalization in random matrix theory.

The OCP describes a plasma of positive charges confined by the potential $V$.
In an alternative interpretation, the effect of the potential is to provide a negative background
charge given by the associated equilibrium measure (described below).
The OCP is also known as Jellium, as Dyson gas, and as (one component) Coulomb gas.
The two-dimensional OCP has fundamental relations to several models in statistical mechanics and probability theory.
For the specific inverse temperature $\beta=1$ (in our units), the OCP is exactly the joint law of the eigenvalues of
a random normal matrix \cite{MR1174692}, and more specifically, for $\beta=1$ and potential $V(z)=|z|^2$
it is the Ginibre ensemble of eigenvalues of a complex Gaussian random matrix \cite{MR0173726,MR2641363}.
For more general values of $\beta$, the OCP also plays a role in the theory of the
Anomalous Quantum Hall Effect, where it arises in the Laughlin wave function \cite{PhysRevLett.50.1395}.

\subsection{Results}
\label{sec:result}

For potentials $V$ that are lower semicontinuous and satisfy the growth condition
\begin{equation} \label{e:Vgrowth}
  \liminf_{|z| \to \infty} \big( V(z) - (2+\epsilon) \log |z| \big) =+\infty
\end{equation}
for some $\epsilon > 0$, 
it is well known that there exists a compactly supported equilibrium measure $\mu_V$
that is the unique minimizer of the convex functional
\begin{equation} \label{e:IV}
  I_V(\mu) = 
  \iint \log \frac{1}{|z-w|} \mu(dz)\,\mu(dw) + \int V(z) \, \mu(dz)
\end{equation}
over the set of probability measures on $\C$; see Theorem~\ref{thm:eqmeasure} below for details.
For $z \in \C^N$, the empirical measure is defined by
\begin{equation*}
  \hat\mu = \frac{1}{N} \sum_{j} \delta_{z_j}.
\end{equation*}
For arbitrary $\beta \in (0,\infty)$,
it is well-known that $\hat\mu \to \mu_V$ vaguely in probability as $N\to\infty$,
with $\hat\mu$ distributed under $P_{N,V,\beta}$;
in fact, a full Large Deviation Principle has been proved \cite{MR1606719,MR1660943}.
Vague convergence concerns the macroscopic behaviour of the systems,
resolving scales of order $1$. The microscopic scale of the system, by which individual
particles are separated, is given by $N^{-1/2}$.

Our first result shows that $\hat\mu \to \mu_V$ also holds on all
mesoscopic scales $N^{-s}$ for any $s \in (0,\frac12)$.
In the random matrix situation, this corresponds to the \emph{local} circular law
\cite{MR2722794,MR3230002},
but here the support of $\mu_V$ is not necessarily a disk.

In the statement of our results below,
$C^k$ denotes the space of (real-valued) $k$-times continuously differentiable functions,
$C^k_c$ the space of functions in $C^k$ which have compact support,
$C^{1,1}$ is the space of differentiable functions whose derivative is Lipschitz continuous,
and $\|f\|_p$ is the standard $L^p$ norm of $f: \C \to \R$ with respect to the 2-dimensional Lebesgue measure on $\C$.

\begin{theorem} \label{thm:locallaw}
Fix $\alpha_0 >0$.
Assume that $V: \C \to \R \cup \{+\infty\}$ obeys \eqref{e:Vgrowth},
is $C^{1,1}$ on a neighbourhood of $\supp \mu_V$,
and satisfies $\alpha_0 \leq \Delta V(z) \leq \alpha_0^{-1}$ for all $z \in \supp \mu_V$. 
Then for any $s \in (0,\half)$,
any $z_0$ in the interior of the support of $\mu_V$ (which we assume to be nonempty),
and for any $f \in C_c^2(\C)$ with support in the disk of radius $\frac12 N^{-s}$ centered at $z_0$,
we have
\begin{equation*}
  \frac{1}{N} \sum_{j=1}^N f(z_j) - \int f(z) \, \mu_V(dz)
  = O\pbb{\pbb{1+\frac{1}{\beta}} \log N}
  \pbb{ N^{-1-2s} \|\Delta f\|_\infty + N^{-\half-s} \|\nabla f\|_{2} } \,,
\end{equation*}
with probability at least $1-\ee^{-(1+\beta)N^{1-2s}}$ for sufficiently large $N$.
The implicit constant depends only on $\alpha_0$, $s$, and $\sup_{\supp \mu_V} |\nabla V|$.%
\footnote{Simultaneously with the first version of this paper,
a result closely related to Theorem~\ref{thm:locallaw} appeared in \cite{1510.01506}.}
\end{theorem}

Theorem~\ref{thm:locallaw} establishes a local density on all scales $N^{-s}$,
$s\in (0,\frac12)$. 
Indeed, by choosing $f$ to be an approximate $\delta$-function,
the theorem implies that, for any $s\in(0,\frac12)$ and any $z \in \supp \mu_V$,
with very high probability,
\begin{equation*}
  \hat \mu(B(z,N^{-s}))
  = \mu_V(B(z, N^{-s})) \pB{1+ O\pB{N^{-\frac12+s+o(1)}}}
  = \mu_V(B(z, N^{-s})) (1+ o(1)),
\end{equation*}
where $B(z,r)$ is the disk of radius $r$ centered at $z$.
Thus the number of particles in $B(z,N^{-s})$
is concentrated around $N\mu_V(B(z,N^{-s})) \approx \frac{1}{4\pi}\Delta V(z) N^{1-2s}$. 
Here the scale $s \in (0,\frac12)$ is optimal.

On the other hand, Theorem~\ref{thm:locallaw} only shows that the fluctuations are at most
as as large as those of a Poisson process (up to a logarithmic correction).
The following theorem improves this bound on the fluctuations significantly,
providing the optimal order of fluctuations for smooth linear statistics.
This shows the particle configurations are \emph{rigid}.

\begin{theorem} \label{thm:rigidity}
Under the assumptions of Theorem~\ref{thm:locallaw},
and assuming in addition that $V$ and $f$ are both in $C^4$,
for any sufficiently small $\varepsilon>0$,
we have
\begin{equation} \label{e:rigidity}
  \sum_{j=1}^N f(z_j) - N\int f(z) \, \mu_V(dz)
  = O\p{N^{\varepsilon}} \pa{\sum_{l=1}^4 N^{-ls} \|\nabla^l f\|_\infty},
\end{equation}
with probability at least $1-\ee^{-\beta N^{\varepsilon}}$ for sufficiently large $N$.
The implicit constant depends on $\varepsilon$, $\alpha_0$, $s$, and $V$.
\end{theorem}

Note that the left-hand side of \eqref{e:rigidity} is \emph{not} normalized by $N$.
Thus Theorem~\ref{thm:rigidity} shows that the fluctuations of smooth linear statistics of the OCP
are $N^{o(1)}$ which is much smaller than the fluctuations of order $N^{\frac12-s}$ for a Poisson process.

\subsection{Strategy}
\label{sec:strategy}

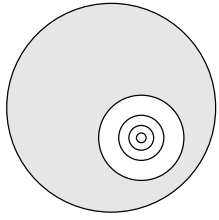
\begin{figure}[t]
\begin{center}
\input{zoom.pspdftex}
\end{center}
\caption{Our strategy involves conditioning on the particle positions outside of increasingly small nested balls.
\label{fig:zoom}}
\end{figure}

Our strategy to prove Theorem~\ref{thm:locallaw} relies on potential theory
and a multiscale iteration of bounds from mean-field theory, as follows.
Using a straightforward mean-field estimate, the density can be bounded on scales larger than $N^{-1/4}$.
In the first iteration, we then fix a small disk $B$ of radius a little bit larger than $N^{-1/4}$ inside the support of the equilibrium measure.
By the initial estimates, $B$ typically contains slightly more than $M = N \times (N^{-1/4})^2 = N^{1/2}$ particles.
We then prove that for most of the particle configurations outside $B$,
we can apply a robust version of the mean-field estimate to the conditional particle distribution inside $B$.
This then yields a density on a smaller scale, namely the new estimate is accurate up to slightly more than $M^{1/2} = N^{1/4}$ particles.
Iterating this procedure, we control scales containing only $N^{\varepsilon}$ particles, for arbitrary $\varepsilon >0$.
A somewhat delicate aspect in this procedure is controlling the properties of effective equilibrium measures
at smaller scales in a sufficiently robust way.

To prove Theorem~\ref{thm:rigidity}, we use the \emph{loop equation}
(which is also known as \emph{Schwinger-Dyson equation} and can also be seen as a \emph{Ward identity}),
with Theorem~\ref{thm:locallaw} as the key input to control the arising error terms. This is more subtle
than in similar applications in one dimension because the resulting equation is singular in two dimensions.
To control the singularity, our proof makes use of the estimates on local scales provided by Theorem~\ref{thm:locallaw};
see, in particular, Lemma~\ref{lem:Ehf}.
Previous uses of the loop equation in related contexts are discussed in Section~\ref{sec:related} below.

\subsection{Related results}
\label{sec:related}

The potential theory associated to \eqref{e:IV} has been the subject of considerable study,
and in fact the comprehensive monograph \cite{MR1485778} is entirely devoted to it. The closely related
obstacle problem is also well studied; see in particular \cite{MR3056295,MR1658612}.

For the positive temperature two-dimensional OCP,
a Large Deviation Principle (LDP) with speed $N^2$
for the empirical measure has been proved \cite{MR1606719,MR1660943}.
As was discussed in Section~\ref{sec:result}, this LDP concerns the macroscopic scale $1$ of the system.
Recently, an LDP with speed $N$
has been proved for a tagged point process \cite{MR3353821,1502.02970}.
Correspondingly, it is shown that for two potentials $V,W$ (satisfying reasonable technical assumptions),
\begin{multline*}
  \log Z_{N,V,\beta} - \log Z_{N,W,\beta} = - \beta N^2 (I_V(\mu_V)-I_W(\mu_W))
  \\
  - \beta N \pa{\frac{1}{\beta}-\frac12}
  \pa{ \int \mu_V(z) \log \mu_V(z) \, dm - \int \mu_W(z) \log \mu_W(z) \, dm }
  + o((1+\beta)N).
\end{multline*}
Related results have also been established in higher dimensions \cite{CPA21570,1502.02970}.
From such estimates, one can obtain a density for the empirical measure down to scale $N^{-1/4}$
(in two dimensions).
To obtain a local density at scale $N^{-1/2+o(1)}$ near any fixed point from bounds on the partition
function, a cancellation with error $N^{o(1)}$ instead of $o(N)$ would be required.
Our approach does not attempt to directly compute the partition function so precisely.
Instead, as sketched in Section~\ref{sec:strategy},
to establish a local density (Theorem~\ref{thm:locallaw}),
we use an inductive scheme 
to improve a weaker estimate with error roughly $O(N\log N)$ and
effectively replace the error by $O(N^{1-2s} \log N)$ on scale $N^{-s}$.
This allows us to reach the optimal scale $s = \frac12 - o(1)$.
In parallel to the first version of this paper, a result closely related to Theorem~\ref{thm:locallaw} has appeared in \cite{1510.01506}.

To improve these estimates and obtain the optimal bound on the
fluctuations in Theorem~\ref{thm:rigidity}, we rely on the \emph{loop equation}.
The loop equation has been used previously to study linear statistics of the OCP and log gases,
in \cite{MR1487983} for any $\beta>0$ in one dimension, and
in \cite{MR3342661,MR2817648} for $\beta=1$ in two dimensions.
Moreover, the loop equation was used to prove rigidity in one dimension in \cite{MR3192527,MR3253704,MR2905803},
and to derive an asymptotic expansion for the partition function in \cite{MR3010191,1303.1045}.

For inverse temperature $\beta=1$, the model is a determinantal point process with explicit correlation kernel,
and using this structure very detailed properties are known; see e.g. \cite{MR2641363}.
In particular, convergence of the fluctuations at macroscopic scale to a Gaussian free field has been established \cite{MR2361453,MR2817648,MR3342661,1507.08674}.
Charge fluctuations have been studied \cite{MR1239571}, and
heuristic arguments for very detailed behaviour are given in \cite{MR2932659,MR2240466,MR1986427,MR2037301}.

On the microscopic scale, for $\beta \neq 1$, simulations suggest
the existence of a phase transition (the critical value $2\beta \approx 140$  is mentioned \cite{JSP28325}),
with the system crystallizing at low temperatures to a so-called Wigner crystal \cite{JSP28325,PhysRev.46.1002}.
The nature of this conjectural phase transition does not appear to be understood well. 
Even in the zero temperature case it is not understood whether the system crystallizes.
Other heuristic predictions are discussed in \cite{MR604143}.

As mentioned previously, a special instance of the two-dimensional one-component plasma is the Ginibre ensemble of random Gaussian
matrices. The natural generalization from the point of view of random matrix theory are
random matrices with i.i.d.\ but non-Gaussian entries (without symmetry constraint).
In this case, Girko's hermitization trick \cite{MR773436} essentially allows to
reduce the problem to that of a symmetric random matrices, and very precise results have been obtained
using this method \cite{MR1437734,MR3230002,MR3230004,MR3278919,MR2409368,MR2722794}.

The $\beta$-ensemble is a one-dimensional version of \eqref{e:Pdef} (with logarithmic interaction also in one dimension),
which is understood extremely well.
In particular, bounds on partition function and global statistics have been proved in \cite{MR1487983,MR1327898},
a large deviation principle in \cite{MR1465640},
a complete $1/N$ expansion for the partition function was derived in \cite{MR3010191,1303.1045},
and universality of local statistics (i.e.\ that these are independent of the potential $V$)
has been proved using the method of orthogonal polynomials for $\beta=1,2,4$ (see e.g.\ \cite{MR1435193,MR1702716}),
and more recently for all $\beta>0$ by direct comparison,
first in \cite{MR3192527,MR3253704,MR2905803,MR3372074}, and then also in \cite{MR3390602,MR3351052,CPA21573}.
Moreover, it has been found that, for general temperature, the $\beta$-ensemble with quadratic potential
can be realized as the joint law of the eigenvalues of a tridiagonal random matrix ensemble \cite{MR1936554}.
Using this representation, the correlation functions have been characterized explicitly in terms of stochastic
differential equations \cite{MR2534097},
and a local version of the semicircle law \cite{MR2966359} has also been proved via this representation.

\subsection{Outline of the paper}

In Section~\ref{sec:eqmeasure}, we state the essential potential theoretic results that underlie our analysis.
In Section~\ref{sec:eqmeasurept}, we prove several estimates on the behaviour of the equilibrium measure under
classes of perturbations of the potential that are important for our analysis.
In Section~\ref{sec:step}, we prove a general estimate that will ultimately be iterated to prove Theorem~\ref{thm:locallaw}.
In Section~\ref{sec:condmeas}, we set-up the conditioning used for the multiscale analysis.
In Section~\ref{sec:pf}, we prove Theorem~\ref{thm:locallaw} by inductively applying the previously
proved results.
Finally, in Section~\ref{sec:rigidity}, we use Theorem~\ref{thm:locallaw} as
the input for the loop equation to prove Theorem~\ref{thm:rigidity}.

\subsection{Notation}

We write $\Delta = \partial_x^2 + \partial_y^2 = \frac14 \bar\partial\partial$ for the Laplacian on $\C$ identified as $\R^2$,
where $\partial = \frac12 (\partial_x - \ii\partial_y)$ and $\bar\partial = \frac12 (\partial_x +\ii\partial_y)$,
$\D = \{z \in \C : |z| < 1\}$ for the open unit disk,
$dm$ for the 2-dimensional Lebesgue measure on $\C$,
and $ds$ for the arclength measure on the boundary of an open subset of $\C$.
There should be no confusion between the measure $ds$ and the scale parameter $s$
as appearing in Theorem~\ref{thm:locallaw}.
The space of (Borel) probability measures on a set $\Sigma \subset \C$ is denoted by $P(\Sigma)$.
We write $(f,\mu) = \int f \, d\mu$ if $\mu$ is a measure and $f \in
L^1(\mu)$, and similarly $(u,v)=\int uv \, dm$ if $u,v \in L^2(\C)$.
We use $C$ and $c$ to denote constants which may change from instance to instance,
and also use the usual Landau notation, $o(x)$ and $O(x)$.
All error estimates are under the tacit assumption that $N$ is sufficiently large.

\section{Characterizations of equilibrium measure}
\label{sec:eqmeasure}

In this short section we describe some standard fundamental results on the equilibrium measure:
its characterization as an energy minimizing measure and as the solution of an obstable problem.
Good references for this material are the monograph \cite{MR1485778} and,
especially with our application in mind, the article \cite{MR3056295}.

\subsection{Energy minimizing measure}

Let $P(\Sigma)$ denote the set of probability measures supported on the closed set $\Sigma \subset \C$.
For $\Sigma \subset \C$, we say that $\Sigma$ has positive (logarithmic) capacity if
\begin{equation*}
  \inf_{\mu \in P(\Sigma)} D(\mu,\mu) < \infty, \qquad \text{where }
  D(\mu,\mu) = I_0(\mu,\mu) = \iint \log \frac{1}{|z-w|} \, \mu(dz) \, \mu (dw).
\end{equation*}
Generally, if some property holds everywhere on $\C$ except on a set of zero (nonpositive) capacity, we say that it holds quasi-everywhere (q.e.).
We remark that a property holding quasi-everywhere implies it almost everywhere (w.r.t.\ Lebesgue measure) but not vice versa.

Throughout this paper, we consider potentials 
that are in the space $C^{1,1}_{loc}(\C)$
of differentiable functions whose derivative is locally Lipschitz continuous,
and which satisfy the growth condition \eqref{e:Vgrowth}.
Our strategy of proof requires us to treat modifications of such potentials where we set the potential
to $\infty$ outside some disk. We will always assume that the set
$
  \Sigma_V = \overline{\{z \in \C: V(z) < \infty\}}
$
has positive capacity. In general the regularity theory of equilibrium measures and their potentials is rather subtle.
However, the assumption $V \in C^{1,1}_{loc}$ (which we need for other reasons as well) considerably simplifies the theory.

For any measure $\mu \in P(\C)$, we denote the (weighted logarithmic) energy of $\mu$ by \eqref{e:IV},
and the (logarithmic) potential of $\mu$ by
\begin{equation*}
  U^{\mu}(z) = \int \log \frac{1}{|z-w|} \, \mu(dw).
\end{equation*}
Since $\Delta \log |\cdot| = 2 \pi \delta_0$ in the sense of distributions, a measure may always be uniquely recovered from its potential.
Conversely, for any superharmonic function $U$ that is harmonic near $\infty$ and satisfies $U(z) \sim \log \frac{1}{|z|}$ as $|z| \to \infty$
there exists a constant $c \in \R$ for which $U+c$ is the potential of some compactly supported $\mu \in P(\C)$.

The following existence and characterization theorem is fundamental.

\begin{theorem}[Frostman] \label{thm:eqmeasure}
  Suppose $V$ is lower semicontinuous and satisfies \eqref{e:Vgrowth}, and that $\Sigma_V$ has positive capacity. Then there exists a unique $\mu_V \in P(\Sigma_V)$ such that
  \begin{equation*}
    I_V(\mu_V) = \inf \{ I_V(\mu): \mu \in P(\Sigma_V) \}.
  \end{equation*}
  The support $S_V = \supp \mu_V$ is compact and of positive capacity, and $I_V(\mu_V) < \infty$.
  
  The energy-minimizing measure $\mu_V$ may be characterized as the unique element of $P(\Sigma_V)$ for which there exists a constant $c \in \R$ such that
  Euler-Lagrange equation
  \begin{align}
    \label{e:EL}
    U^{\mu_V} + \tfrac12 V = c & \quad \text{q.e.\ in $S_V$} \quad \text{and}\\
    U^{\mu_V} + \tfrac12 V \geq c & \quad \text{q.e.\ in $\C$} \nonumber
  \end{align}
  holds. Also, necessarily $c = F_V$, with the definition $F_V = I_V(\mu_V) - \frac12 (V,\mu_V)$.
\end{theorem}

\begin{proof}
  See \cite[Theorem I.3.3]{MR1485778}.
\end{proof}

\subsection{Obstacle problem}

Based on the characterization \eqref{e:EL}, to determine the equilibrium measure $\mu_V$ it is essentially enough to determine its support $S_V$.
However, changing a measure $\mu$ locally generally changes its potential $U^\mu$ everywhere, making the determination of $S_V$ through \eqref{e:EL} difficult.
The characterization of $\mu_V$ as the energy-minimizing measure is likewise non-local and thus difficult to apply to the problem of actually determining $S_V$.
To get hold of $S_V$ in a local, effective way we will instead apply the characterization of $U^{\mu_V}$ as the solution
of an obstacle problem associated to $V$. This connection is discussed for example in \cite{MR3056295}, to which we will refer in this section.

Denote the class of subharmonic functions on $\C$ by $\r{subh}(\C)$ and, given $V$, define
\begin{equation} \label{e:obstacleproblem}
  u_V(z) = \sup \left\{ v(z): v \in \r{subh}(\C), \, v \leq \tfrac12 V \text{ q.e. on } \C, \,
    \limsup_{|z| \to \infty} \big(v(z) - \log |z|\big) < \infty \right\}.
\end{equation}
Note that the conditions $v \in \r{subh}(\C)$ and $\limsup_{|z| \to \infty} \big(v(z) - \log |z|\big) < \infty$ imply
that $v$ is of the form $c - U^\nu$ for some $c \in \R$ and positive measure $\nu$ with $\|\nu\| \leq 1$.
Further observe that $F_V - U^{\mu_V}$ with $F_V$ and $\mu_V$ as in \eqref{e:EL} satisfies all the three requirements for $v$
and thus $F_V - U^{\mu_V} \leq u_V$ quasi-everywhere. The converse inequality is given in Theorem \ref{thm:eqmeasure-regularity},
giving the promised characterization of the equilibrium potential.

Denote 
\begin{equation*}
  S_V^* = \{ z \in \C: u_V(z) \geq \tfrac{1}{2} V(z) \}.
\end{equation*}
Up to a set of capacity zero, $S_V^*$ is the same as the set $\{ u_V(z) = \tfrac{1}{2} V(z) \}$ and so it is called the \emph{coincidence set}.
The precise relation of the obstacle problem to the energy minimizing problem is given in the
following theorem, summarizing several results from \cite{MR3056295}.

\begin{theorem} \label{thm:eqmeasure-regularity}
  Let $V$ be as in Theorem \ref{thm:eqmeasure} and define $u_V$ by \eqref{e:obstacleproblem}. Then
\begin{enumerate}
\item
For q.e.\ $z \in \C$,
\begin{equation*}
  u_V(z) = F_V - U^{\mu_V}(z).
\end{equation*}
Especially $u_V(z) = \tfrac{1}{2} V(z)$ q.e.\ in $S_V$, which in turn implies $S_V \subset S_V^*$.
\item
The measure $\mu_V$ is given by
\begin{equation}
  \label{e:equilibriumdensity}
  \mu_V = \tfrac{1}{2\pi} \Delta u_V ,
\end{equation}
where the Laplacian is understood in the distributional sense.
\item
Suppose $V$ is $C^{1,1}$ in a neighbourhood of $z \in S_V^*$.
Then also $u_V$ is $C^{1,1}$ in a neighbourhood of $z$. Also if $z \notin S_V^*$ then $u_V$ is harmonic in a neighbourhood of $z$.
\end{enumerate}
\end{theorem}

\begin{proof}
Proofs of (i) and (ii) can for example be found in
\cite[Proposition~3.2 and Corollary~3.4]{MR3056295}. For (iii), see \cite[Theorem 2]{MR1658612}.
\end{proof}

In particular, (ii--iii) imply that if $V \in C^{1,1}$ in a neighbourhood of $S_V$,
then the equilibrium measure $\mu_V$ is absolutely continuous with respect to the Lebesgue measure.

\section{Perturbations of equilibrium measure}
\label{sec:eqmeasurept}

Using the characterizations of the equilibrium measure of Theorems~\ref{thm:eqmeasure}--\ref{thm:eqmeasure-regularity},
we prove estimates on its dependence under
certain classes of perturbations of the potential. These estimates will play an important role in the proof of Theorem~\ref{thm:locallaw}.
Throughout this section, $V: \C \to \R \cup \{+\infty\}$ is a potential that is locally in $C^{1,1}$
and satisfies the assumptions of Theorem~\ref{thm:eqmeasure}, and $\mu_V$ is the  associated
equilibrium measure according to Theorem~\ref{thm:eqmeasure}.

\subsection{Local perturbation}

Given a potential $V$, 
the next result concerns the change
of its equilibrium measure and its energy under a change $V \to V-f$, where $f$ is a small local perturbation.

\begin{proposition} \label{prop:eqmeasf}
  Let $f \in C^2(\C)$ be bounded and satisfy  the conditions
  $\supp \Delta f \subset S_V$ and $\Delta V \geq \Delta f$ in $S_V$.
  Then $\mu_{V-f} = \mu_V - \frac{1}{4\pi} \Delta f$. In particular,
  $S_{V-f} \subseteq S_V$,
  \begin{equation} \label{e:zetafeq}
    U^{\mu_{V-f}} + \frac12 (V-f) - F_{V-f} =
    U^{\mu_{V}} + \frac12 V - F_{V},
  \end{equation}
  and 
  \begin{equation} \label{e:energyfineq}
    I_{V-f}(\mu_{V-f}) = I_V(\mu_V) - (f,\mu_V) - \frac{1}{8\pi} (f,-\Delta f).
  \end{equation}
\end{proposition}

In preparation of the proof of the proposition we note that,
for $f$ bounded, twice differentiable and with compact $\supp \Delta f$,
as in the statement of the proposition,
it follows that
\begin{equation} \label{e:flogint}
f(z) = U^{-\frac{1}{2 \pi} \Delta f}(z) + c = \int \log \frac{1}{|z-w|} \left(- \frac{1}{2 \pi} \Delta f(w)\right) \, dm + c
\end{equation}
for some $c \in \R$, i.e.\ that $f$ can be written as a constant plus the logarithmic potential of its Laplacian.
Namely, the difference $f - U^{-\frac{1}{2 \pi} \Delta f}$ satisfies
\begin{equation*}
\Delta (f - U^{-\frac{1}{2 \pi} \Delta f})(z)
= \Delta f(z) - \int \left(-\frac{1}{2 \pi} \Delta f(w)\right) \, d (-2 \pi \delta_z(w)) = 0
\end{equation*}
for all $z \in \C$, implying that it is a harmonic function.
Clearly also $|(f - U^{-\frac{1}{2 \pi} \Delta f})(z)| = O(\log |z|)$ as $|z| \to \infty$.
By (a strong version of) Liouville's theorem for harmonic functions the difference is thus constant.
In particular, by the representation \eqref{e:flogint}, we have
\begin{equation} \label{e:Deltafint}
  \int \Delta f \, dm = 0,\quad
  \nabla f(z) = O(1/|z|^2) \;\; \text{ as $|z|\to\infty$},\quad
  (f,-\Delta f) = \|\nabla f\|_2^2.
\end{equation}

\begin{proof}
  Let $\mu = \mu_V - \frac{1}{4\pi} \Delta f$.
  By the assumption $\Delta V \geq \Delta f$ in $S_V$ and since $\mu_V = \frac{1}{4\pi} \Delta V$ on $S_V$,
  $\mu$ is a positive measure. 
  By \eqref{e:Deltafint}, we have
  \begin{equation*}
    \int d\mu  = 1 - \frac{1}{4\pi} \int \Delta f \, dm = 1,
  \end{equation*}
  which means that $\mu \in P(\C)$.
  Moreover, \eqref{e:flogint} implies
  \begin{equation} \label{e:zetacal}
    U^{\mu} + \frac12 (V-f)
    = U^{\mu_V}
    - \frac{1}{4\pi} \int \log\frac{1}{|z-w|} \Delta f(w) \, m(dw) 
    + \frac12 V - \frac12 f
    = U^{\mu_V} + \frac12 V - \frac12 c,
  \end{equation}
  where $c$ is the same constant as in \eqref{e:flogint}.
  By Theorem~\ref{thm:eqmeasure} applied with potential $V$,
  the right-hand side of the above equality is equal to $F_V - \tfrac12 c$ in $S_V$ and at least $F_V - \tfrac12 c$ outside $S_V$.
  The uniqueness statement of Theorem~\ref{thm:eqmeasure} applied with the potential
  $V-f$ now implies that $\mu_{V-f} = \mu$, as claimed. 
  It also follows that $F_{V-f} = F_V - \tfrac12 c$, and therefore \eqref{e:zetacal} implies \eqref{e:zetafeq}.

  It remains to show \eqref{e:energyfineq}.
  Since $\mu_{V-f} =\mu_V - \frac{1}{4\pi}\Delta f$, we indeed have
  \begin{align*}
    I_{V-f}(\mu_{V-f}) - I_V(\mu_V) + (f,\mu_V)
    &= 
    D(\tfrac{1}{4\pi} \Delta f,\tfrac{1}{4\pi} \Delta f)
    + (f,\tfrac{1}{4\pi} \Delta f)
    - 2 D(\mu_V,\tfrac{1}{4\pi} \Delta f) - (V, \tfrac{1}{4\pi} \Delta f)
    \\
    &=
    \tfrac{1}{8\pi} (f,-\Delta f) - \tfrac{1}{4\pi} (f,-\Delta f)
    =
    -\tfrac{1}{8\pi} (f,-\Delta f),
  \end{align*}
  where we have used that the last two terms on the first line cancel,
  by integration by parts and since $d\mu_V = \frac{1}{4\pi} \Delta V \, dm$ in the interior of its support.
  Thus \eqref{e:energyfineq} holds as claimed.
\end{proof}

\subsection{Restriction}

The following proposition shows the important property that, given a potential $V$ with equilibrium measure $\mu_V$,
the potential $W$ defined by adding the logarithmic potential of the charge of $\mu_V$ contained in some region $B^c$ to $V$,
has equilibrium measure $\mu_W$ given simply by the rescaled restriction of $\mu_V$ to $B$.

\begin{proposition} \label{prop:condeqmeas}
  Let $B \subset S_V$ a compact subset and set
  \begin{equation*}
    W(z) = \frac{1}{\mu_V(B)} \left( V(z) + 2\int_{S_V \setminus B} \log \frac{1}{|z-w|} \, \mu_V(dw) \right).
  \end{equation*}
  Then $S_W = B$ and
  \begin{equation} \label{e:muWideal}
    \mu_W = \frac{1}{\mu_V(B)}  \mu_V|_{B},
  \end{equation}
  where $\mu_V|_B$ is the restriction of $\mu_V$ to $B$.
\end{proposition}

\begin{proof}
  Define $\mu$ by the right-hand side of \eqref{e:muWideal}. Then
  \begin{align*}
    U^\mu(z) + \frac{1}{2} W(z)
    &= \frac{1}{\mu_V(B)}
      \left( \int_B \log \frac{1}{|z-w|} \mu_V(dw) + \frac{1}{2} V(z) 
      +  \int_{S_V \setminus B} \log\frac{1}{|z-w|} \mu_V(dw) \right)
    \nonumber \\
    &\begin{cases}
      = \frac{F_V}{\mu_V(B)} & \text{for q.e. } z \in B
      \\
      \geq \frac{F_V}{\mu_V(B)} & \text{for q.e. } z \in \C
    \end{cases}
  \end{align*}
  by the Euler--Lagrange equation associated to $V$.
  Since, by definition, $\mu$ is a probability measure,
  Theorem~\ref{thm:eqmeasure} implies that $\mu = \mu_W$.
\end{proof}

\subsection{Harmonic perturbation}

In the following, we consider a class of perturbations of the potential $V$ that are 
harmonic inside the support of the equilibrium measure.

For convenience, we assume here that $S_V = \rho\overline{\D}$ for some $\rho>0$, where $\D$ is the open unit disk.
(This will be sufficient for our application in the proof of Theorem~\ref{thm:locallaw} with general potential; see Section~\ref{sec:pf}.)
Furthermore, we assume that $\frac{1}{4 \pi} \Delta V \geq \alpha$ in $\rho\D$ for some parameter $\alpha>0$.
The class of perturbations $W$ is as follows.
Let $\nu$ be
a positive measure with $\supp \nu \cap \rho\D = \emptyset$,
$R \in C(\rho\overline{\mathbb{D}})$ be harmonic in $\rho\D$, and $t > 0$.
Then $W$ is given by
\begin{equation} \label{e:Winnerrad}
W(z) = \begin{cases}
tV(z) + 2 U^\nu(z)
+ 2R(z), & z \in \rho\overline{\mathbb{D}} \\
\infty, & z \in \rho\mathbb{D}^* \,.
\end{cases}
\end{equation}
Both perturbations $U^\nu$ and $R$ are harmonic inside $\rho\D$. We will later assume that $R$ is small,
in a certain sense, while $U^\nu$ is allowed to be more singular but generated by a \emph{positive} measure $\nu$.

We write $\mathbb{D}^* = \C \setminus \overline{\mathbb{D}}$ for the open complement of the unit disk.
Moreover, for $z \in \partial \rho\D$, we write $\bar n = \bar n(z) = z/|z|$ for the outer unit normal, and
\begin{equation*}
  \partial_n^-f(z) = \lim_{\varepsilon \downarrow 0} \frac{f(z)-f(z-\varepsilon \bar n)}{\varepsilon} 
\end{equation*}
for the derivative in the direction $\bar n$ taken from inside $\rho\D$.

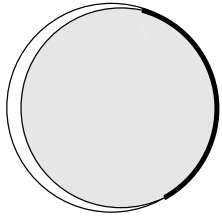
\begin{figure}[t]
\begin{center}
\input{innerradius.pspdftex}
\end{center}
\caption{For a class of perturbations $W$ of $V$ as above \eqref{e:Winnerrad},
the support of $\mu_{W}$ contains the bulk of the support of $\mu_V$,
and the density on the boundary remains bounded.
\label{fig:innerradius}}
\end{figure}

The next two propositions show that the bulk of the equilibrium measure $\mu_V$
is stable under suitable pertubations $W$ of the form \eqref{e:Winnerrad},
and that the density of $\mu_W$ on the boundary remains bounded.
Both properties are illustrated in Figure~\ref{fig:innerradius}.

\begin{proposition} \label{prop:innerradius}
Suppose that $V$ and $W$ are as above \eqref{e:Winnerrad}.
Then we have
\begin{equation}
\label{e:kappabd}
S_W \supset \left\{ z \in \rho\D: \dist(z, \rho\D^*) \geq \kappa \right\},
\quad
\text{where } \kappa = 4 \sqrt{\frac{\max(\|\nu\|,2 \rho\|\partial_n^- R\|_{\infty,\partial \rho\D} + (t-1))}{\alpha t}}.
\end{equation}
\end{proposition}

\begin{proposition} \label{prop:bddensity}
Suppose that $V$ and $W$ are as above \eqref{e:Winnerrad} and assume in addition
that $\mu_V$ is absolutely continuous with respect to the 2-dimensional Lebesgue measure.
Then $\mu_W = \mu + \eta$, where $\mu$ is absolutely continuous with respect to $\mu_V$, and $\eta$ absolutely continuous
with respect to the arclength measure $s$ on $\partial \rho\mathbb{D}$ with the Radon--Nikodym derivative bounded by
\begin{equation} \label{e:etabd}
\rho
\normB{\frac{d \eta}{ds}}_\infty 
\leq 
\frac{1}{2\pi} \pB{ \|\eta\| + \|\nu\| + 2\rho\|\partial_n^- R\|_{\infty,\partial \rho\D} + |1-t|\rho\|\partial_n^- V\|_{\infty,\partial\rho\D} }.
\end{equation}
\end{proposition}

\begin{remark}
The bounds \eqref{e:kappabd} and \eqref{e:etabd} are effective
for small $\nu$ and $R$, and $t$ close to $1$.
For larger perturbations, the bounds still remain valid (but sometimes vacuously).
For example, for $\rho=1$ note that $S_V = \overline{\D}$ implies that $\alpha \leq 1/\pi$.
As $t\to\infty$, we have $\kappa \to 4/\sqrt{\alpha} \geq 4/\sqrt{\pi} \geq 1$, so $S_W = \emptyset$ as expected. 
Suppose $\nu=0$ and $R=0$. Then $\kappa = 4 \sqrt{\max(0,t-1)/(\alpha t)}$. This is $0$ for $t\leq 1$ and increasing for $t\geq 1$,
also as expected.
\end{remark}

\subsubsection{Proof of Proposition~\ref{prop:innerradius}}

As preparation for the proof of Proposition~\ref{prop:innerradius} we first recall the behaviour of the distributional Laplacian
for functions with a discontinuous gradient on a curve and then prove a technical lemma.

Let $\gamma$ be a smooth Jordan curve with interior domain $D^+$ and exterior domain $D^-$. Suppose $f \in C^2(D^+ \cup D^-) \cap C(\mathbb{C})$
and further that $\|\nabla f\|_\infty < \infty$ in a neighbourhood of the curve $\gamma$. Then the distributional Laplacian of $f$ coincides
with the usual pointwise Laplacian off the curve $\gamma$ and on $\gamma$ it is the measure $(\partial^+ - \partial^-) f \, ds$,
where $\partial^+$ and $\partial^-$ denote the normal derivatives from the outside and inside, respectively,
taken at a point of $\gamma$ and $ds$ is the arclength measure on $\gamma$. Concisely we may write
\begin{equation} \label{e:Deltaf}
\Delta f = \Delta f \, dm + (\partial^+ f - \partial^- f) \, ds,
\end{equation}
where the left-hand side denotes the Laplacian understood in the distributional sense and on the right-hand side $dm$ is the area measure
and $ds$ the arclength measure on $\gamma$. This formula can be deduced from Green's identity as follows.

Let $\phi \in C^\infty_c(\mathbb{C})$ be a test function whose support intersects both $D^+$ and $D^-$.
To determine the distribution $\Delta f$, first write
\begin{equation*}
\int \Delta f \, \phi = \int f \, \Delta \phi \, dm = \int_{D^+} f \, \Delta \phi \, dm + \int_{D^-} f \, \Delta \phi \, dm,
\end{equation*}
where the first equality is by definition of the distributional derivative and the second holds by the continuity of $f$ and $\phi$
and the smoothness of $\gamma$. Again by the smoothness assumptions, Green's identity may be applied twice to both terms separately to obtain
\begin{equation*}
\int_{D^\pm} f \, \Delta \phi \, dm
= \int_{\partial D^\pm} f \, \partial_n \phi \, ds - \int_{\partial D^\pm} \partial_n f \, \phi \, ds + \int_{D^\pm} \Delta f \, \phi \, dm
\end{equation*}
where $\partial_n$ denotes the outer normal derivative in the corresponding domain.
Finally, note that $\nabla \phi$ is $0$ on $\partial \supp \phi$ and that on $\gamma$ the outer normals of $D^+$ and $D^-$ are
negatives of each other, by the continuity of $f$ implying that
$
\int_{\partial D^+} f \, \partial_n \phi + \int_{\partial D^-} f \, \partial_n \phi = 0
$.
Summing up,
\begin{align*}
\int \Delta f \, \phi
& = - \int_{\partial D^+} \partial_n f \, \phi \, ds + \int_{D^+} \Delta f \, \phi \, dm - \int_{\partial D^-} \partial_n f \, \phi \, ds + \int_{D^-} \Delta f \, \phi \, dm \\
& = \int_{D^+} \Delta f \, \phi \, dm + \int_{D^-} \Delta f \, \phi \, dm + \int_\gamma (\partial^+ f - \partial^- f) \, \phi \, ds.
\end{align*}
Since the test function $\phi$ is arbitrary, this is equivalent to \eqref{e:Deltaf}.

The following lemma contains the central idea of the proof of Proposition \ref{prop:innerradius}.
While checking the calculus of the lemma is slightly tedious, the idea is clearly illustrated in Figure \ref{fig:lfit}:
the lemma shows that the dotted graph lies below the solid graph.
%
%
% l(x,r) = abs(x) <= r ? (0.5 - log(r) - (x**2)/(2*r**2)) : -log(abs(x))
% plot -log(abs(x-8)), l(x-5,1)-1.65
% set terminal fig
% set output "lfit.fig"
% plot -log(abs(x-8)), l(x-5,1)-1.65

\begin{figure}[t]
\begin{center}
\input{lfit.pspdftex}
\end{center}
\caption{
The figure illustrates (a one-dimensional projection of)
the construction of the test function \eqref{e:vdef} for a single external charge located at $w \not\in\D$,
in the case $\tilde R=0$ and $t=1$.
The density of the equilibrium measure imposes a lower bound on $r$. Then if $z_0$ is sufficiently far from $w$, we can find
$\tilde z$ and $k$ such that $\log\frac{1}{|z-w|}$ and $l_r(z-\tilde z)+k$ match at $z_0$
and $\log\frac{1}{|z-w|}$ dominates $l_r(z-\tilde z)+k$ everywhere.
Lemma~\ref{lem:lrsingle} shows that the dotted graph lies below the solid graph.
\label{fig:lfit}}
\end{figure}
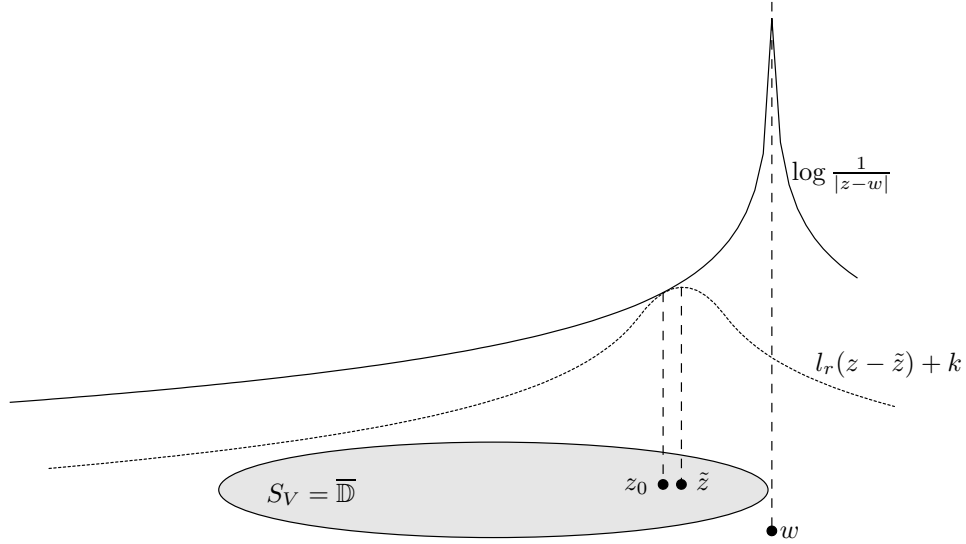

To state the lemma,
for any $r>0$, define the function $l_r : \C \to \R$ by
\begin{equation} \label{e:lrdef}
l_r(z) = (-\log { |\cdot|} * \, \frac{1}{\pi r^2} \ind{B(0,r)})(z) = \begin{cases}
\frac{1}{2} + \log \frac{1}{r} - \frac{|z|^2}{2 r^2} , & |z| \leq r \\
\log \frac{1}{|z|}, & |z| \geq r.
\end{cases}
\end{equation}
For later use, we note that
\begin{equation}
  \label{e:Deltal}
  \nabla l_r(z) = \frac{-z}{r^2 \vee |z|^2}, \qquad
   \Delta l_r(z) = - \frac{2}{r^2} \ind{B_r(0)}
  .
\end{equation}

\begin{lemma} \label{lem:lrsingle}
Let $z_0 \in \C, w \in \C$, $\sigma \geq \tfrac{1}{2}$ and $r \in (0,1)$ be given so that $|z_0-w| \geq 2r$.
Then there exist $\tilde{z} \in \mathbb{C}$ and $k \in \mathbb{R}$ such that
\begin{equation} \label{e:log-matching}
\sigma \big(l_r(z_0-\tilde{z}) + k\big) = \frac{1}{2} \log \frac{1}{|z_0-w|}
\quad \textrm{and} \quad
\sigma \big(l_r(z-\tilde{z}) + k\big) \leq \frac{1}{2} \log \frac{1}{|z-w|} \quad \textrm{for all } z \in \mathbb{C}.
\end{equation}
Moreover, the point $\tilde z$ lies on the line passing through $z_0$ and $w$ at distance at most $r$ from $z_0$
between $z_0$ and $w$.
\end{lemma}

\begin{proof}
First, we choose $\tilde{z} \in \C$ and $k\in \R$ so that
\begin{equation} \label{e:choicetildezk}
\sigma \nabla l_r(z_0-\tilde{z}) = \frac{1}{2} \left. \nabla \log \frac{1}{|z-w|} \right|_{z=z_0} = \frac{1}{2} \frac{z_0-w}{|z_0-w|^2}
\quad \textrm{and} \quad
\sigma \big( l_r(z_0-\tilde{z}) + k \big) = \frac{1}{2} \log \frac{1}{|z_0-w|}
.
\end{equation}
To see that this is possible, first note that 
$\frac{1}{2} \left| \frac{z_0-w}{|z_0-w|^2}\right| \leq \frac{1}{4 r}$ by the assumption $|z_0-w| \geq 2r$.
By \eqref{e:Deltal}, the map $z \mapsto \sigma \nabla l_r(z_0-z)$ takes the disk
$B_r(z_0)$ bijectively onto $B_{\sigma/r}(0) \supset B_{1/(4r)}(0)$.
It follows there exists a unique choice of $\tilde z \in B_r(z_0)$ so that the gradients of
$\sigma l_r(\cdot - \tilde{z})$ and $\tfrac{1}{2} \log \tfrac{1}{|\cdot - w|}$ match at $z_0$.
The second equality can then be arranged by the choice of $k$.

It remains to be shown that with the choice \eqref{e:choicetildezk} it is in fact true that
\begin{equation} \label{e:lbdC}
\sigma \big(l_r(z-\tilde{z}) + k\big) \leq \frac{1}{2} \log \frac{1}{|z-w|} \quad \textrm{for all } z \in \mathbb{C}.
\end{equation}
Clearly, the point $\tilde{z}$ lies between the points $z_0$ and $w$ on the line $\cal L$ connecting these two points.
We will first prove that the inequality in \eqref{e:lbdC} holds for $z \in \cal L$.
For the proof, it is helpful to keep Figure~\ref{fig:lfit} in mind.
Without loss of generality assume $w = 0$ and $z_0 >0, \tilde{z} > 0$ so that $\cal L$ coincides with $\R$.
Thus it needs to be shown that
\begin{equation*}
f(x) := \frac{1}{2} \log \frac{1}{|x|} \geq \sigma \big( l_r(x-\tilde{z}) + k \big) =: g(x), \quad x \in \R,
\end{equation*}
where $\tilde{z}$ is chosen as in \eqref{e:choicetildezk}.
Denote by $h$ the common tangent of the graphs of $f$ and $g$ drawn at $x=z_0$.
Since $f$ is convex and $g$ is concave on $[\tilde{z}-r,\tilde{z}+r]$,
the graph of $f$ lies above $h$ and the graph of $g$ lies below $h$ on this interval.
Especially $g(x) \leq f(x)$ on $[\tilde{z}-r,\tilde{z}+r]$. Moreover, since $f'(x) < 0$ and $g'(x) > 0$ for $x \in (0,\tilde{z})$,
the inequality $g(x) \leq f(x)$ holds by these observations for $x \in (0,\tilde{z}+r]$.
To prove the inequality for $x \in [\tilde{z}+r,\infty)$ note that $g'(t) = -\frac{\sigma}{t-\tilde{z}} \leq -\frac{1}{2t} = f'(t)$
for $t \in [\tilde{z}+r,\infty)$. It follows that
\begin{equation*}
g(x) - g(\tilde{z}+r) = \int_{\tilde{z}+r}^x g'(t) \, dt \leq \int_{\tilde{z}+r}^x f'(t) \, dt = f(x) - f(\tilde{z}+r),
\end{equation*}
which by $g(\tilde{z}+r) \leq f(\tilde{z}+r)$ implies the desired inequality $g(x) \leq f(x)$, now proven for $x \in (0,\infty)$.
For $x \in (-\infty,0)$ it also holds that $g'(x) \leq f'(x)$ and it is clear that $f(x) \geq g(x)$ as $x \to 0^-$,
so it remains to check the inequality as $x \to -\infty$. For $|x|$ large we write the difference $f-g$ as
\begin{equation*}
\frac{1}{2} \log \frac{1}{|x|} - \sigma \left( \log \frac{1}{|x-\tilde{z}|} + k \right) \leq \frac{1}{2} \log \left( 1 - \frac{|\tilde{z}|}{|x|}\right) - \sigma k,
\end{equation*}
and from this form it is clear that $g(x) \leq f(x)$ on the whole negative real axis if and only if $k < 0$. This can be verified by the calculation
\begin{equation*}
\sigma k = \frac{1}{2} \log \frac{1}{z_0} - \sigma l_r(z_0-\tilde{z}) < \frac{1}{2} \log \frac{1}{z_0} - \sigma l_r(r) \leq \frac{1}{2} \log \frac{1}{2r} - \sigma \log \frac{1}{r} = \left( \frac{1}{2} - \sigma \right) \log \frac{1}{r} - \frac{1}{2} \log 2 < 0.
\end{equation*}

\begin{figure}[h]
\begin{center}
\input{linebd.pspdftex}
\end{center}
\caption{
The circle centered at $\tilde z$ is the level set of the left-hand side of \eqref{e:lbdC},
while the circle centered at $w$ is the level set of the right-hand side.
Therefore is suffices to verify $\text{LHS}(a) \leq \text{RHS}(b)$.
\label{fig:levelsets}}
\end{figure}
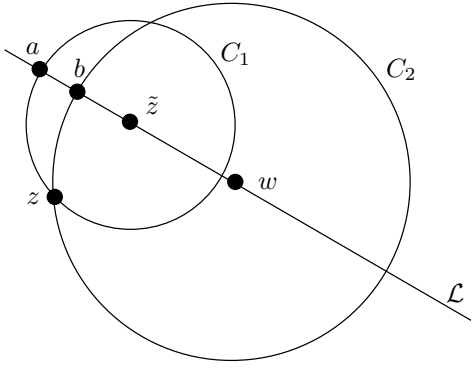

It remains to extend the inequality \eqref{e:lbdC} from the line $\cal L$ passing through $w$ and $z_0$ to the whole plane.
This is easiest done by inspection of the level sets of the left- and right-hand sides of \eqref{e:lbdC} as in Figure~\ref{fig:levelsets}.
The level sets of the left-hand side are the circles centered at $\tilde{z}$ and the level sets of the right-hand side are the circles centered at $w$.
Let $z \in \mathbb{C}$ be arbitrary and consider the circles $C_1$ and $C_2$ that pass through $z$ and have centres $\tilde{z}$ and $w$ respectively.
Let $a$ denote the point of intersection of $C_1$ and $\cal L$ for which $\tilde{z}$ lies between $a$ and $w$,
and let $b$ denote the point of intersection of $C_2$ and $\cal L$ for which $b$ and $\tilde{z}$ are in the same direction as seen from $w$.
The left-hand side of \eqref{e:lbdC} a decreasing function of the distance $|z-\tilde{z}|$,
so it suffices to show that the point $b$ lies within the circle $C_1$. If $\tilde{z}$ lies between $b$ and $w$, then we have
\begin{equation*}
|b-w| = |z-w| \leq |\tilde{z}-w| + |z-\tilde{z}| = |\tilde{z}-w| + |a-\tilde{z}| = |a-w|,
\end{equation*}
which implies $|b-\tilde{z}| \leq |a-\tilde{z}|$ since $\tilde{z}$ lies between $b$ and $w$ and also between $a$ and $w$. Thus if $\tilde{z}$ lies between $b$ and $w$ we have proven \eqref{e:lbdC}. In the remaining case $b$ lies between $\tilde{z}$ and $w$. Then
\begin{equation*}
|b-w| + |\tilde{z}-a| = |z-w| + |z-\tilde{z}| \geq |\tilde{z}-w|,
\end{equation*}
which implies that $a$ does not lie strictly between $b$ and $\tilde{z}$,
again implying that $b$ lies inside $C_1$.
The inequality \eqref{e:lbdC} has now been proven for all $z \in \C$,
completing the proof of the lemma.
\end{proof}

\begin{proof}[Proof of Proposition~\ref{prop:innerradius}]
As a preliminary step we show that we can assume $\rho=1$ without loss of generality.
Indeed, assume that $S_V = \rho\overline{\D}$, and set $V_\rho(z) = V(\rho z)$ and $W_\rho(z) = W(\rho z)$.
From Theorem~\ref{thm:eqmeasure} it then follows that
$\mu_{V_\rho}(dz) = \mu_{V}(\rho \, dz)$ and $\mu_{W_\rho}(dz) = \mu_{W}(\rho \, dz)$,
so in particular $S_{V_\rho} = \overline{\D}$.
Then the claim with $\rho=1$ implies that
$S_{W_\rho} \supset \{z \in \D: \dist(z,\D^*) \geq \kappa_\rho \}$, with
\begin{equation}
  \sqrt{\rho^2 \alpha}\kappa_\rho 
  =
  4 \sqrt{\frac{\max(\|\nu\|,2\rho\|\partial_n^- R\|_{\infty,\partial \rho\mathbb{D}} + (t-1))}{t}},
\end{equation}
and the original claim then follows by rescaling.

From now on, we therefore assume that $\rho=1$.
Let
\begin{equation*}
D = \left\{ z \in \mathbb{D} : \dist(z, \mathbb{D}^*) \geq \kappa \right\}.
\end{equation*}
We will show that $D \subset S_W^*$ by exhibiting, for every $z_0 \in D$, a function $v = v_{z_0}$ that satisfies $v(z_0) = \frac{1}{2}W(z_0)$
and the requirements in \eqref{e:obstacleproblem} with $V$ replaced by $W$.
Thus we have $u_W = \frac12 W$ in $D$, and \eqref{e:equilibriumdensity} and the fact that $W-tV$ is harmonic in $\D \supset D$
imply that $\mu_W = \frac{1}{ 4\pi} t \Delta V$ in $D$. In particular $D \subset \supp \mu_W$ as claimed.

We may without loss of generality assume that $\nu \neq 0$. Namely, if $\nu = 0$ it follows from the statement
of the proposition applied with $\nu$ replaced by $\nu = \varepsilon \delta_M$ (with $\varepsilon > 0$ small and $M > 1$ large)
and $R$ replaced by $R - U^{\nu'}$
that the estimate holds also for $\nu = 0$.

The main difficulty of the proof is in giving $v=v_{z_0}$ for every $z_0 \in D$ and checking that the requirements are satisfied.
Fix $z_0 \in D$.
Let $u_V$ be the solution to the obstacle problem \eqref{e:obstacleproblem} with the original potential $V$.
Define the functions $l_r$ for $r>0$ as in \eqref{e:lrdef} and define
\begin{equation}  \label{e:Rtildedef}
\tilde{R}(z) = \begin{cases}
R(z), & z \in \overline{\mathbb{D}} \\
R(1/\bar{z}), & z \in \mathbb{D}^*.
\end{cases}
\end{equation}
Observe that $\tilde{R}$ is bounded and continuous on $\mathbb{C}$ and harmonic in $\mathbb{D} \cup \mathbb{D}^*$.
(Though we will not need the fact, it is the unique such function.)
We will choose $v = v_{z_0}$ as a function of the form
\begin{equation} \label{e:vdef}
v(z) = t u_V(z) + \sigma L(z) + \tilde{R}(z) + \gamma G(z), \quad L(z) = \int \left( l_r\big(z-\tilde{z}(w)\big) + k(w) \right) \nu(dw)
,
\end{equation}
where $G(z) = \max(0,\log |z|)$ is the Green's function of $\mathbb{D}^*$ with pole at $\infty$,
and $\sigma > 0$, $r > 0$, $\alpha > 0$, $k: \supp \nu \to \R$ and $\tilde{z}: \supp \nu \to \mathbb{D}$ are parameters. Set
\begin{equation*}
\gamma = 2 \|\partial_n^- R\|_{\infty,\partial \mathbb{D}}, \quad
\sigma = \max\left(\frac12,\frac{\gamma+(t-1)}{\|\nu\|}\right),
\quad \textrm{and} \quad
r = 2 \sqrt{\frac{\|\nu\| \sigma}{\alpha t}} = \frac12 \kappa.
\end{equation*}
The functions $\tilde{z}$ and $k$ will be determined through Lemma~\ref{lem:lrsingle} later;
see also Figure~\ref{fig:lfit}.

We first verify the conditions on $v$ required in \eqref{e:obstacleproblem}:
that $\limsup_{|z| \to \infty} \big( v(z) - \log |z| \big) < \infty$,
that $v$ is subharmonic on $\C$, and that $v \leq \frac12 W$ on $\C$.
The asymptotics as $|z| \to \infty$
of the different terms in \eqref{e:vdef} are
$u_V(z) \sim t \log |z|$,
$\sigma L(z) \sim -\sigma \|\nu\| \log |z|$,
$\tilde{R}(z) \sim 1$
and $\gamma G(z) \sim \gamma \log |z|$.
Thus the growth condition at $\infty$ in \eqref{e:obstacleproblem} is satisfied,
as $\gamma + t - \sigma \|\nu\| \leq 1$ by the definitions of the parameters.
Next, we show that $v$ is subharmonic.
By the symmetry of the definition \eqref{e:Rtildedef},
the jump of the gradient of $\tilde{R}$ at $\zeta \in \partial \mathbb{D}$ is $2 \partial_n^- R(\zeta)$.
Thus by \eqref{e:Deltal} and \eqref{e:Deltaf} the (distributional) Laplacian of $v$ is given by
\begin{equation*}
\Delta v
= t\Delta u_V + \sigma \Delta L + \Delta \tilde{R} + \gamma \Delta G
= \frac{t}{2} \Delta V \ind{\mathbb{D}}  - \sigma \int \frac{2}{r^2} \ind{B_r(\tilde{z}(w))} \, \nu(dw) + 2 \partial_n^- R \, ds + \gamma \, ds,
\end{equation*}
where $ds$ is the arclength measure on $\partial \mathbb{D}$.
The points $\tilde{z}(w)$ will be chosen so that $B_r(\tilde{z}(w)) \subset \mathbb{D}$ for all $w \in \supp \nu$,
so the subharmonicity of $v$ in the interior of $\mathbb{D}$ follows from
\begin{equation*}
\sigma \int \frac{2}{r^2} \ind{B_r(\tilde{z}(w))} \, \nu(dw)
\leq 
\frac{2\sigma  \|\nu\| }{r^2} \ind{\mathbb{D}}
\leq \frac{t}{2} \alpha \ind{\mathbb{D}} \leq \frac{t}{2} \Delta V \ind{\mathbb{D}}.
\end{equation*}
by the definition of $r$.
On the other hand, 
on $\partial \mathbb{D}$ the positivity of $\Delta v$ follows from 
$2 \partial_n^- R + \gamma \geq 0$ by our choice of $\gamma$. 

By Theorem~\ref{thm:eqmeasure-regularity}~(i) and the assumption $S_V = \overline{\D}$,
we have $S_V^* \supseteq S_V = \overline{\D}$, and therefore 
\begin{equation*}
u_V(z) = \frac{1}{2} V(z) \quad \textrm{for all } z \in \overline{\D}.
\end{equation*}
By definition we also have
\begin{equation*}
  \tilde{R}(z) = R(z) \quad \textrm{for all } z \in \overline{\D}.
\end{equation*}
Moreover, we have $G(z) = 0$ for all $z \in \overline{\D}$.
To have $v(z) \leq \frac{1}{2} W(z)$ for all $z \in \overline{\D}$ and $v(z_0) = \frac{1}{2} W(z_0)$ as required,
it thus suffices to show that we can choose the parameters $\tilde{z}: \supp \nu \to \mathbb{D}$ and $k: \supp \nu \to \R$
in the definition of $L(z)$ such that
\begin{align}
\label{e:Lz}
\sigma L(z) 
&\leq \frac{1}{2} \int \log \frac{1}{|z-w|} \nu(dw)
\quad \textrm{for all } z \in \overline{\D},
\\
\label{e:Lz0}
\sigma L(z_0)
&= \frac{1}{2} \int \log \frac{1}{|z_0-w|} \nu(dw).
\end{align}
To achieve this, for every $w \in \supp \nu$ we apply Lemma \ref{lem:lrsingle} to obtain $\tilde{z}(w)$ and $k(w)$ for which
$$
\sigma \big( l_r(z_0-\tilde{z}(w)) + k(w) \big) = \frac{1}{2} \log \frac{1}{|z_0-w|}  \quad \textrm{and} \quad  \sigma \big( l_r(z-\tilde{z}(w)) + k(w) \big) \leq \frac{1}{2} \log \frac{1}{|z-w|} \quad \textrm{for } z \in \mathbb{C}.
$$
The assumptions of the lemma are in force, as $\sigma \geq \tfrac{1}{2}$ and $|z_0-w| \geq 2r = \kappa$ by our definitions and assumptions,
and further if $r \geq 1$, $\kappa = 2r \geq 2$ and there is nothing to prove. (It is also obvious from the proof of Lemma \ref{lem:lrsingle}
that the maps $w \mapsto \tilde{z}(w)$ and $w \mapsto k(w)$ are measurable.) The requirements \eqref{e:Lz}--\eqref{e:Lz0} are now satisfied,
finishing the proof.
\end{proof}

\subsubsection{Proof of Proposition~\ref{prop:bddensity}}

In the proof of Proposition~\ref{prop:bddensity}, we will use the following
general formula for logarithmic potentials.
Let $\gamma \subset \C$ be a $C^1$ curve and $\eta$ a measure supported on $\gamma$
for which the potential $U^\eta$ is continuous on $\C$. Then for $z \in \gamma$ we have
\begin{equation} \label{e:dUeta}
\partial_n^- U^\eta(z) = \pi \lim_{r \to 0^+} \frac{\eta(B_r(z))}{s(B_r(z))} - \int_\gamma \frac{z-w}{|z-w|^2} \cdot \bar{n} \, \eta(dw),
\end{equation}
where $\partial_n^-$ denotes a one-sided derivative
in the normal direction $\bar{n} = \bar{n}(z)$ and $s$ denotes the arclength measure of $\gamma$,
if the limit on the right-hand side exists. 
In addition we will use the closely related fact that
that for every $z \in \supp \eta$ it holds that
\begin{equation} \label{e:bdd-limsup}
\limsup_{r \to 0^+} \frac{\eta(B_r(z))}{s(B_r(z))} = \infty 
\quad \textrm{if and only if} \quad
\limsup_{\varepsilon \to 0^+} \frac{1}{\varepsilon} \left( U^\eta(z) - U^\eta(z-\varepsilon \bar{n}) \right) = \infty.
\end{equation}
For completeness, we provide a proof of \eqref{e:dUeta} and \eqref{e:bdd-limsup} below the proof of the proposition.

\begin{proof}[Proof of Proposition~\ref{prop:bddensity}]
As in the proof of Proposition~\ref{prop:innerradius}, without loss of generality, we can assume $\rho=1$.

Let $\mu$ be the absolutely continuous part of $\mu_W$ and set $\eta = \mu_W - \mu$. Write $d \eta = \frac{d \eta}{ds} \, ds + d \eta_s$, where $\eta_s$ is singular with respect to the arclength measure.
Given Theorem \ref{thm:eqmeasure-regularity}, it is to be shown that $\eta_s \equiv 0$ and that $\frac{d \eta}{ds}$ obeys \eqref{e:etabd}.

First,  by \cite[Theorems I.4.8 and I.5.1]{MR1485778}, the potential $U^{\mu_W}$ is continuous on the full plane.
(Indeed, by Theorem~I.4.8, continuity is clear inside $\D$ as well as outside $S_V$.
For $z_0 \in \partial S_V \cap \partial \D$, condition (iv)' of Theorem I.5.1 is satisfied since $\C \setminus \Sigma =\D^*$
and any point on its boundary $\partial \D$ is thus a regular point for the Dirichlet problem.)
Therefore, by \eqref{e:dUeta}, for every $z \in \supp \eta$ for which the measure $\eta$ is differentiable with respect to $s$, we have
\begin{equation} \label{e:bdd-normal1}
\pi \lim_{r \to 0^+} \frac{\eta(B_r(z))}{s(B_r(z))} = \partial_n^- U^{\mu_W}(z) + \int \frac{z-w}{|z-w|^2} \cdot \bar{n} \, \mu_W(dw)
.
\end{equation}
We will show that the right-hand side of \eqref{e:bdd-normal1}
is bounded by the right-hand side of \eqref{e:etabd}.
In fact the same argument shows that $\limsup_{r \to 0^+} \frac{\eta(B_r(z))}{s(B_r(z))} < \infty$ at every $z \in \supp \eta$.
From this it follows that $\eta_s \equiv 0$, and then, since $\lim_{r \to 0^+} \frac{\eta(B_r(z))}{s(B_r(z))} = \frac{d \eta}{ds}(z)$
for $ds$-almost every $z \in \supp \eta$ by \cite[Theorem~3.22]{MR1681462},
we see that the bound \eqref{e:etabd} holds.

By the continuity of the potential $U^{\mu_W}$ equality in the Euler--Lagrange equation \eqref{e:EL}
holds at every (rather than quasi-every) $z \in \supp \eta$,
i.e., $U^{\mu_W}(z) = c - \frac{1}{2} W(z)$, and
we also have $U^{\mu_W}(z) \geq c - \frac{1}{2} W(z)$ everywhere. 
Thus
\begin{equation*}
  U^{\mu_W}(z) + \frac12 W(z) - (U^{\mu_W}(z-\varepsilon \bar n) + \frac12 W(z-\varepsilon \bar n)) \leq 0,
\end{equation*}
and therefore, assuming that the derivative exists at $z$,
\begin{equation} \label{e:bdd-normal2}
\partial_n^- U^{\mu_W}(z) \leq - \frac{1}{2} \partial_n^- W(z)
= - \frac{t}{2} \partial_n^- V(z) + \int \frac{z-w}{|z-w|^2} \cdot \bar{n} \, \nu(dw) - \partial_n^- R(z).
\end{equation}
Combining \eqref{e:bdd-normal1} and \eqref{e:bdd-normal2} gives
\begin{equation} \label{e:bdd-bound}
\pi \lim_{r \to 0^+} \frac{\eta(B_r(z))}{s(B_r(z))}
\leq - \frac{t}{2} \partial_n^- V(z) + \int \frac{z-w}{|z-w|^2} \cdot \bar{n} \, \nu(dw) - \partial_n^- R(z) + \int \frac{z-w}{|z-w|^2} \cdot \bar{n} \, (\mu+\eta)(dw)
\end{equation}
for those $z \in \supp \eta$ for which the limit exists.

\begin{figure}[h]
\begin{center}
\input{nbd.pspdftex}
\end{center}
\caption{
For $w$ in the halfplane $H^-$ we have $(z-w)\cdot \bar n < 0$
and $(z-w)\cdot \bar n > 0$ for $w \in H^+$.
The bound $\frac{z-w}{|z-w|^2} \cdot \bar{n}\leq \frac12$ holds for all $z \in \C \setminus  \D$.
\label{fig:nbd}}
\end{figure}
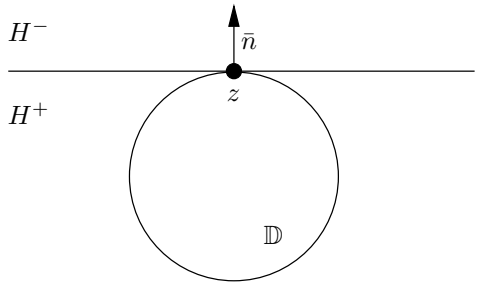

Given $z \in \partial\D$, denote $H^+ = \{ w: (z-w) \cdot \bar{n} > 0 \}$ and $H^- = \{ w: (z-w) \cdot \bar{n} < 0 \}$.
These sets are the half-planes into which the tangent of $\partial \mathbb{D}$ at $z$ divides the plane,
$H^+$ being the one that contains $\mathbb{D}$;
see Figure~\ref{fig:nbd}.
Observe the inequality $\frac{z-w}{|z-w|^2} \cdot \bar{n}\leq \frac12$ in $\overline{H^+} \setminus \mathbb{D}$. 
Indeed, supposing $z=(1,0)$ for notational convenience, $\overline{H^+} \setminus \mathbb{D} = \{(x,y): x \leq 1, y^2 \geq 1-x^2 \}$ and thus
$$
\frac{z-w}{|z-w|^2} \cdot \bar{n} = \frac{1-x}{(x-1)^2+y^2} \leq \frac{1-x}{2-2x} = \frac{1}{2} \quad \textrm{for } w = (x,y) \in \overline{H^+} \setminus \mathbb{D}.
$$
By the definition of $H^-$ thus also $\frac{z-w}{|z-w|^2} \cdot \bar{n}\leq \frac12$ in $\C\setminus \mathbb{D}$.
Since $\supp \nu \subset \C \setminus \mathbb{D}$ and $\supp \eta \subset \partial \mathbb{D} \subset \overline{H^+} \setminus \mathbb{D}$ it follows that
$$
\int \frac{z-w}{|z-w|^2} \cdot \bar{n} \, \nu(dw) \leq \frac{1}{2} \|\nu\| \quad \textrm{and} \quad \int \frac{z-w}{|z-w|^2} \cdot \bar{n} \, \eta(dw) \leq \frac{1}{2} \|\eta\|.
$$
To bound the integral with respect to $\mu$ in \eqref{e:bdd-bound} we note that, since $S_V = \overline{\D}$ and $V \in C^{1,1}_{loc}$,
and since $\mu_V$ is absolutely continuous with respect to the Lebesgue measure, we have $\mu = \ind{S_W} \mu_V$
(recall the perturbation in harmonic in $\D$, so can only change the support but not the density of the equilibrium measure)
and thus
$$
\int \frac{z-w}{|z-w|^2} \cdot \bar{n} \, \mu(dw) \leq \int \frac{z-w}{|z-w|^2} \cdot \bar{n} \, \mu_V(dw) = -\partial_n^- U^{\mu_V}(z) = \frac{1}{2} \partial_n^- V(z).
$$
With these estimates \eqref{e:bdd-bound} gives
$$
\pi \lim_{r \to 0^+} \frac{\eta(B_r(z))}{s(B_r(z))}
\leq
\frac{1-t}{2} \partial_n^- V(z) +
\frac{1}{2} \|\nu\| - \partial_n^- R(z) + \frac{1}{2} \|\eta\|
$$
as needed.

Finally, by replacing the one-sided derivative with the corresponding upper limit of difference quotients,
the bound \eqref{e:bdd-normal2} holds for all $z \in \supp \eta$.
The claim that $\limsup_{r \to 0^+} \frac{\eta(B_r(z))}{s(B_r(z))} < \infty$ at every $z \in \supp \eta$ then follows from \eqref{e:bdd-limsup}.
\end{proof}

\begin{proof}[Proof of \eqref{e:dUeta}]
Suppose $\gamma$ is such that $\gamma:[a,b] \to \C$ is an arclength parametrization, $\gamma(c) = z$ and let $\eta^\star$ be the pullback of $\eta$ on $[a,b]$. We may suppose $c \in (a,b)$, $z = 0$ and $\bar{n} = (0,-1)$.
Thus $\gamma(c+\varepsilon) = (\varepsilon+o_1(\varepsilon),o_2(\varepsilon))$.
Let $\phi(\varepsilon)$ be a function for which $\phi(\varepsilon)/\varepsilon \to \infty$, $\phi(\varepsilon) \to 0$ as $\varepsilon \to 0$.
Then
\begin{align*}
\frac{1}{\varepsilon} \left( U^\eta(z) - U^\eta(z - \varepsilon\bar{n}) \right) = & \frac{1}{\varepsilon} \int_\gamma \left( \log \frac{1}{|z-w|} - \log \frac{1}{|(z-\varepsilon\bar{n})-w|} \right) \eta(dw) \\
= & \frac{1}{\varepsilon} \int_{\gamma \setminus \gamma([c-\phi(\varepsilon),c+\phi(\varepsilon)])} \left( \log \frac{1}{|z-w|} - \log \frac{1}{|(z-\varepsilon\bar{n})-w|} \right) \eta(dw) \\
& + \frac{1}{\varepsilon} \int_{-\phi(\varepsilon)}^{\phi(\varepsilon)} \frac{1}{2} \log \frac{(s+o_1(s)^2)+(\varepsilon+o_2(s)^2)^2}{(s+o_1(s))^2 + o_2(s)^2} \eta^\star(ds).
\end{align*}
Noting that the integrand is not singular since $\bar{n}$ is the normal of $\gamma$ at $z$, the first term tends to $- \int_\gamma \frac{z-w}{|z-w|^2} \cdot \bar{n} \, \eta(dw)$ as $\varepsilon \to 0^+$. In the latter term we make the change of variables $s \to \varepsilon s$, note that the terms involving $o_1$ and $o_2$ may be dropped and get
\begin{equation*}
\frac{1}{2} \int_{-\phi(\varepsilon)/\varepsilon}^{\phi(\varepsilon)/\varepsilon} \log \left( 1 + \frac{1}{s^2} \right) \frac{\eta^\star(\varepsilon\,ds)}{\varepsilon} \stackrel{\varepsilon \to 0^+}{\longrightarrow} \pi \lim_{r \to 0^+} \frac{\eta^\star([-r,r])}{2r} = \pi \lim_{r \to 0^+} \frac{\eta(B_r(z))}{s(B_r(z))}
\end{equation*}
if the last limits exist,
by the evaluation $\int_{-\infty}^\infty \log (1 + \frac{1}{s^2}) \, ds = 2 \pi$.
The claim \eqref{e:bdd-limsup} follows by estimating the integrand from below by an indicator function.
\end{proof}

\section{Large deviation estimate for one step}
\label{sec:step}

The result of this section is the following estimate.
It follows from elementary estimates from mean-field theory:
Newton's electrostatic theorem and Jensen's inequality;
see e.g.\ \cite{MR0401029}.

The statement of our assumptions below is somewhat involved due to the required generality
that the equilibrium measure may have charge on the boundary of a domain.
In the special case that there is no boundary charge, the assumptions can be significantly simplified.
Moreover, our conditions on the boundary are not the most general possible,
but convenient to verify in our application of interest (in which the domain will be a small disk).
For example, we only use assumption (ii) to control the logarithmic potential of $v$,
and the exponent $1/2$ could be replaced by any exponent greater than $0$ or directly by an estimate on the logarithmic potential.

\begin{proposition} \label{prop:step}
Let $\Sigma = \Sigma_W$ be a smooth domain.
Given a potential $W \in C^{1,1}_{loc}(\Sigma_W)$ possibly depending on the number of particles $M$,
assume that there exist $u: \Sigma_W \to \R_+$ and $v: \partial \Sigma_W \to \R_+$ such that
$d\mu_W = u \, dm + v \, ds$, where $dm$ is the 2-dimensional Lebesgue measure and $ds$ is the arclength measure on $\partial \Sigma_W$.
We further make the following assumptions.
\begin{enumerate}
\item $\Delta W \geq \alpha$ on $S_W$ for some $\alpha > 0$ (which is allowed to depend on $M$);

\item there exist $A_u,A_v$ such that $\|\ind{S_W}\Delta W\|_{\infty} \leq C M^{A_u}$, $\|v\|_\infty \leq C M^{A_v}$ and
\begin{equation} \label{e:Ai}
  \sup_{w \in \C}
  \int \frac{v(z) \, s(dz)}{\sqrt{|z-w|}} \leq CM^{A_v};
\end{equation}

\item there exist $A_\nabla,A_n$ and an exceptional set $F \subset \Sigma_W$ such that $\|\ind{S_W \cap F^c} \nabla W\|_\infty \leq C M^{A_\nabla}$ and
such that we have
$\sup_{t \in [0,M^{-A_n}]} \int_{F \cap \partial \Sigma_W} |W(z-t\bar n) - W(z)| \, v(z) \, s(dz) \leq CM^{-1}$;

\item
there exist $A_\zeta, A_n$ such that 
\begin{equation} \label{e:ind-zeta-ass}
  \int_\C \ee^{-M\beta \zeta(z)} \, m(dz) \leq CM^{A_\zeta}, \qquad 
  \text{where }
  \zeta = 2U^{\mu_W}+ W - F_W
  ,
\end{equation}
and that for any $w \in \supp v \subset \partial \Sigma_W$
it holds that 
$w - \varepsilon \bar{n}_w \in  \Sigma_W$ for all $\varepsilon < M^{-A_n}$
where $\bar{n}_w$ denotes the outward normal of $\partial \Sigma_W$.
\end{enumerate}
Then for any bounded $f \in C^2(\C)$ with $\supp \Delta f \subset S_W$, 
and for any $\xi \geq 1+1/\beta$,
\begin{equation} \label{e:step}
  \left|\frac{1}{M} \sum_j f(z_j) - \int f \, d\mu_W\right|
  = O(\xi) \pa{ \frac{M \log M}{\alpha M^2} \|\Delta f\|_\infty + \left(\frac{M \log M}{M^2} \right)^{1/2}  \|\nabla f\|_2 }
  ,
\end{equation}
with probability at least $1-\ee^{-\xi\beta M \log M}$ under $P_{M,W,\beta}$.
The implicit constant depends only on the numbers $A$ and $C$ assumed above.

In fact, if there is no boundary charge ($v=0$), the condition that $\Delta f$ has support in $S_V$
can be extended to the condition that $\Delta f$ has compact support in $\Sigma$.
\end{proposition}

\begin{remark} \label{rem:step}
  Proposition~\ref{prop:step} gives an estimate at scales $M^{-s}$, $s<1/4$.
  Indeed, let $f_0 \geq 0$ be a smooth function supported in the unit ball
  with $\int f_0 \, dm = 1$, and set $f_s(z) = M^{2s} f(M^{s}z)$. Then
  \begin{equation*}
    \|\Delta f_s\|_\infty = O(M^{4s}), \quad
    \|\nabla f_s\|_2 = (f_s,-\Delta f_s) = O(M^{4s}),
  \end{equation*}
  and thus the error bounds in \eqref{e:step} are
  \begin{equation*}
    M^{-1} (\log M) M^{4s}, \quad
    M^{-1/2} (\log M)^{1/2} M^{2s},
  \end{equation*}
  which go to $0$ if $s<1/4$.
\end{remark}

\subsection{Bounds on partition function}

For $\beta >0$ and $f: \C \to \R$, we define the perturbed partition function 
\begin{equation*}
  Z(f) = \int_{\C^M} \ee^{-\beta H_{M,W}(z) + \sum_{j} f(z_j)} \, m^{\otimes M}(dz).
\end{equation*}
This is simply the partition function associated to the potential $W - \frac{1}{\beta M} f$.
In the following two lemmas, upper and lower bounds are shown, from which Proposition~\ref{prop:step}
then follows easily.

\begin{lemma} \label{lem:ind-lb}
  Let $f \in C^2(\C)$ be bounded with $\supp \Delta f \subset S_W$
  and $|\Delta f| \leq \beta M \Delta W$ in $\supp \Delta f$,
  and assume that the conditions in Proposition~\ref{prop:step} hold.
  Then
  \begin{equation*}
    -\frac{1}{\beta} \log Z(f)
    \geq M^2 I_W(\mu_W) - \frac{1}{\beta} M \int f \, d\mu_W - \frac{1}{{8}\pi\beta^2} (f,-\Delta f)
    - C M \log M
    - \frac{1}{\beta} A_\zeta M \log M
  \end{equation*}
  for a constant $C > 0$ only depending on $A_u$, $A_\nabla$ and $A_v$.

  In fact, if there is no boundary charge ($v=0$), the condition that $\Delta f$ has support in $S_V$
  can be extended to the condition that $\Delta f$ has compact support in $\Sigma$.
\end{lemma}

\begin{proof}
Let $\rho = \frac{1}{\pi} \ind{\D}$.
Since $\rho$ is a radially symmetric probability distribution, Newton's electrostatic theorem states that
\begin{equation} \label{e:Newton}
  \log \frac{1}{|z|} \geq \int \log \frac{1}{|z-w|} \, \rho(w) \, m(dw).
\end{equation}
In particular, if $\mu^{(\varepsilon)} = \rho_\varepsilon * \mu$ where $\rho_\varepsilon(z) = \varepsilon^{-2} \rho(\varepsilon^{-1}z)$ for $\rho$ as above,
\begin{equation} \label{e:Onsager0}
  D(\delta_z, \delta_w) \geq D(\delta_z^{(\varepsilon)}, \delta_w^{(\varepsilon)}).
\end{equation}
It follows that
\begin{equation*}
H_{M,W}(z) - \frac{1}{\beta} \sum_j f(z_j) \geq M^2 D(\hatmueps,\hatmueps)
- M D(\delta_0^{(\varepsilon)},\delta_0^{(\varepsilon)}) + M^2(W - \tfrac{1}{\beta M}f,\hat\mu).
\end{equation*}
Let $\mu_{W,f}$ be the equilibrium measure for the potential $W - \frac{1}{\beta M} f$, and write
\begin{align*}
D(\hatmueps,\hatmueps) + (W-\tfrac{1}{\beta M} f,\hat\mu)
&=  D(\mu_{W,f},\mu_{W,f}) + (W-\tfrac{1}{\beta M} f,\mu_{W,f}) \\
&\qquad + D(\hatmueps-\mu_{W,f},\hatmueps-\mu_{W,f}) + 2 (U^{\mu_{W,f}},\hatmueps-\hat\mu) \\
&\qquad + 2 (U^{\mu_{W,f}} + \tfrac{1}{2} ( W - \tfrac{1}{\beta M} f ),\hat\mu-\mu_{W,f}) 
\end{align*}
to get
\begin{align*}
H_{M,W}(z) - \frac{1}{\beta} \sum_j f(z_j)
&\geq M^2 \left( D(\mu_{W,f},\mu_{W,f}) + (W-\tfrac{1}{\beta M} f,\mu_{W,f}) \right)
- M D(\delta_0^{(\varepsilon)},\delta_0^{(\varepsilon)}) \\
&\qquad + M^2 \left( D(\hatmueps-\mu_{W,f},\hatmueps-\mu_{W,f}) + 2(U^{\mu_{W,f}},\hatmueps-\hat\mu) \right)\\
&\qquad + M^2 \left( 2 (U^{\mu_{W,f}} + \tfrac{1}{2} ( W - \tfrac{1}{\beta M} f ),\hat\mu-\mu_{W,f}) \right).
\end{align*}
The first term on the second line is nonnegative since $D$ is positive definite for signed measures with total measure $0$.
By Proposition~\ref{prop:eqmeasf},
the first term on the third line equals
$\int \zeta(z) \hat\mu(dz)$ (recall the definition of $\zeta$ in \eqref{e:ind-zeta-ass}),
which by the Euler--Lagrange equation 
is identically $0$ on $S_W$ and positive elsewhere. Therefore
\begin{equation} \label{e:ind-lb-H}
H_{M,W}(z) - \frac{1}{\beta} \sum_j f(z_j) \geq M^2 h 
- M D(\delta_0^{(\varepsilon)},\delta_0^{(\varepsilon)}) 
+ M^2 \int \zeta(z) \hat\mu(dz) + 2 M^2 (U^{\mu_{W,f}},\hatmueps-\hat\mu)
,
\end{equation}
where we abbreviate $h = I_{W-\frac{1}{\beta M} f}(\mu_{W-\frac{1}{\beta M} f})$.

The asymptotics of the second term in \eqref{e:ind-lb-H}
are given by $M D(\delta_0^{(\varepsilon)},\delta_0^{(\varepsilon)}) = M (\log \frac{1}{\varepsilon}+O(1))$.
To estimate the term $2 M^2 (U^{\mu_{W,f}},\hatmueps-\hat\mu)$ in \eqref{e:ind-lb-H},
we use the equality $U^{\mu_{W,f}} =  U^{\mu_{W}} + \frac{1}{2\beta M} f + c$ by Proposition~\ref{prop:eqmeasf}. 
From the explicit formula for $\rho_\varepsilon * (-\log) = l_\varepsilon$ in \eqref{e:lrdef} it follows that
\begin{align*}
|(U^{\mu_W})^{(\varepsilon)}(z) - U^{\mu_W}(z)|
&= \left| \int \left( l_\varepsilon(z-w) - \log \frac{1}{|z-w|} \right) \mu_W(dw) \right|
\leq \int \log \frac{\varepsilon}{|z-w| \wedge \varepsilon} \, \mu_W(dw) \\
& \leq \sqrt{\varepsilon} \int_{B_\varepsilon(z)} \frac{\mu_W(dw)}{\sqrt{|z-w|}}.
\end{align*}
where we used $\log t \leq \sqrt{t}$ for $t \geq 1$. 
Write $\mu_W = u \, dm + v \, ds$. Then by assumption (ii) we have
\begin{equation*}
\sqrt{\varepsilon} \int_{B_\varepsilon(z)} \frac{u(w) \, m(dw)}{\sqrt{|z-w|}} \leq C M^{A_u} \varepsilon^2,
\qquad
\sqrt{\varepsilon} \int_{B_\varepsilon(z)} \frac{v(w) \, s(dw)}{\sqrt{|z-w|}} \leq  C M^{A_v} \sqrt{\varepsilon}.
\end{equation*}
For $f \in C^2(\C)$ with $|\Delta f| \leq \beta M \Delta W \leq \beta M^{1+A_u}$, using the radial symmetry of $\rho$, we also have
\begin{equation*}
  \tfrac{1}{2\beta M} |(f^{(\varepsilon)}-f,\hat\mu)|
  \leq \tfrac{1}{\beta M} C \varepsilon^2 \|\Delta f\|_\infty
  \leq C \varepsilon^2 M^{A_u}.
\end{equation*}
By the associative property of convolution, for any two measures $\mu$ and $\nu$, we have
$(U^\mu,\nu^{(\varepsilon)}) = ((U^{\mu})^{(\varepsilon)},\nu)$, from which we conclude that
\begin{equation*}
|(U^{\mu_{W,f}},\hatmueps-\hat\mu)|
\leq
C \varepsilon^2 M^{A_u} + C \sqrt{\varepsilon} M^{A_v} .
\end{equation*}
We may thus take $\varepsilon = M^{-\max(A_u/2,2A_v)}$ so that for $M$ large \eqref{e:ind-lb-H} reads
\begin{equation} \label{e:ind-lb-H2}
H_{M,W}(z) - \frac{1}{\beta} \sum_j f(z_j) \geq M^2 h 
+ M^2 \int \zeta(z) \hat\mu(dz) - C M \log M,
\end{equation}

By assumption (iv) of Proposition~\ref{prop:step} and \eqref{e:ind-lb-H2},
\begin{align*}
Z(f)
&= \int_{\C^M} \ee^{-\beta H_{M,W}(z) + \sum_{j} f(z_j)} \, m^{\otimes M}(dz)\\
& \leq \int_{\C^M} \ee^{-\beta \left( M^2 h 
    + 2 M^2 \int \zeta(z) \hat\mu(dz) - C M \log M \right)  } \, m^{\otimes M}(dz) \\
& = \ee^{-\beta \left( M^2 h 
    - C M \log M \right)} \left( \int_\C \ee^{-\beta M \zeta(z)} \, m(dz) \right)^M \\
& \leq \ee^{-\beta \left( M^2 h 
    - C M \log M \right)} M^{M A_\zeta}.
\end{align*}
Finally, by Proposition~\ref{prop:eqmeasf},
$
h \geq I_W(\mu_W) - \tfrac{1}{\beta M} (f,\mu_W) - \frac{1}{{8} \pi \beta^2 M^2}(f,-\Delta f)
$,
which gives the desired estimate.

In the case that $v=0$ the estimates can be extended to $f$ with $\Delta f$ compactly supported in $\Sigma$
as follows. We replace $\mu_{W,f}$ by $\mu_W$ in all estimates and 
instead of $D(\hat \mu^{(\varepsilon)} - \mu_{W,f}, \hat \mu^{(\varepsilon)} - \mu_{W,f}) \geq 0$ we use
integration by parts, the Cauchy--Schwarz inequality and
$-|ab|+|b|^2 \geq -|a|^2/4$ to get
\begin{align*}
  - \tfrac{1}{\beta M}(f, \hat \mu^{(\varepsilon)}) + D(\hat \mu^{(\varepsilon)} - \mu_W, \hat \mu^{(\varepsilon)} - \mu_W)
  &= -\tfrac{1}{\beta M} (f, \mu_W) + \tfrac{1}{2\pi} (\tfrac{1}{\beta M} \nabla f + \nabla U^{\hat\mu^{(\varepsilon)}-\mu_W}, \nabla U^{\hat\mu^{(\varepsilon)}-\mu_W})
  \\
  &\geq -\tfrac{1}{\beta M} (f, \mu_W) - \tfrac{1}{8\pi \beta^2 M^2} (f,-\Delta f).
\end{align*}
The integration by parts is justified if $v=0$.
This leads to the same bound as previously without using that $\supp \Delta f  \subset S_W$.
\end{proof}

\begin{lemma} \label{lem:ind-ub}
Assume the conditions in Proposition \ref{prop:step} hold. Then
\begin{equation} \label{e:Z-ub}
-\frac{1}{\beta} \log Z(0) \leq M^2 I_W(\mu_W) + C \frac{1}{\beta} M \log M
\end{equation}
for a constant $C > 0$ that depends only on the constants $A$.
\end{lemma}

\begin{proof}
For $w \in \supp v$, we set $E_\varepsilon(w) = \{ w-t\bar{n}_w: t \in [0,\varepsilon] \}$.
As in assumption~(iv), we choose $\varepsilon > 0$ small enough
such that for any $w \in \supp v \subset \partial \Sigma_W$
it holds that  $E_\varepsilon(w) \subset \Sigma_W$,
and define the measure $\mu_W^{(\varepsilon)}$ by
\begin{equation*}
d \mu_W^{(\varepsilon)} = u \, dm +
dv^{(\varepsilon)}, \quad
v^{(\varepsilon)}(dz) = \int_{\supp v} v(w) \frac{1}{\varepsilon} s_{E_\varepsilon(w)}(dz) \, s(dw),
\end{equation*}
where $s_{E_\varepsilon(w)}$ denotes the arclength measure on the segment $E_\varepsilon(w)$ and $s$
the arclength measure on $\partial \Sigma$.
For $\varepsilon < M^{-A_n}$ observe that $v^{(\varepsilon)}$ is absolutely continuous with respect to $dm$ and that
its density (which by a slight abuse of notation we also denote by $v^{(\varepsilon)}$, so that $v^{(\varepsilon)}(dz) = v^{(\varepsilon)}(z) \, dm$)
is bounded by $O(\varepsilon^{-1}\|v\|_\infty$) by assumption (iv).
Denote the density of $\mu_W^{(\varepsilon)}$ at $z$ by $\mu_W^{(\varepsilon)}(z)$ and apply
Jensen's inequality to see that
\begin{align*}
Z(0) &= \int_{\C^M} \ee^{-\beta H_{M,W}(z)} \, m^{\otimes M}(dz)\\
& \geq \int_{\C^M} \ee^{-\beta H_{M,W}(z) - \sum_j \log \mu_W^{(\varepsilon)}(z_j)} \prod_j \mu_W^{(\varepsilon)}(z_j) \, m^{\otimes M}(dz) \\
& \geq \exp\left( \int_{\C^M} \left( -\beta H_{M,W}(z) - \sum_j \log \mu_W^{(\varepsilon)}(z_j) \right) \prod_j \mu_W^{(\varepsilon)}(z_j) \, m^{\otimes M}(dz) \right) \\
& = \exp\left( - \beta M(M-1) I_W(\mu_W^{(\varepsilon)}) - M \int \log \mu_W^{(\varepsilon)}(z) \, \mu_W^{(\varepsilon)}(dz) \right)
\end{align*}
Thus
\begin{equation} \label{e:ub-1}
-\log Z(0) \leq \beta M^2 I_W(\mu_W^{(\varepsilon)}) + M \int \log \mu_W^{(\varepsilon)}(z) \, \mu_W^{(\varepsilon)}(dz).
\end{equation}
By the assumptions of Proposition \ref{prop:step} and the definition of $\mu_W^{(\varepsilon)}$, the density $\mu_W^{(\varepsilon)}(\cdot)$ is bounded by $C M^{A_u} + C \tfrac{1}{\varepsilon} M^{A_v}$ and therefore
\begin{equation} \label{e:ub-entropy}
M \int \log \mu_W^{(\varepsilon)}(z) \, \mu_W^{(\varepsilon)}(dz) \leq C M + M \log \frac{1}{\varepsilon} + \max(A_u,A_v) M \log M.
\end{equation}
The energy may be estimated as follows. For the interaction energy we write
\begin{align*}
D(\mu_W^{(\varepsilon)},\mu_W^{(\varepsilon)}) - D(\mu_W,\mu_W)
= (U^{\mu_W^{(\varepsilon)}},\mu_W^{(\varepsilon)}) - (U^{\mu_W},\mu_W)
&= (U^{\mu_W^{(\varepsilon)}} - U^{\mu_W},\mu_W^{(\varepsilon)} + \mu_W)
\\
&= (U^{v^{(\varepsilon)}} - U^{v},\mu_W^{(\varepsilon)} + \mu_W).
\end{align*}
From
\begin{align*}
U^{v^{(\varepsilon)}}(z)
= \int \log \frac{1}{|z-\xi|} v^{(\varepsilon)}(d\xi)
= \int v(w) \, s(dw) \, \int_0^\varepsilon \log \frac{1}{|z-(w-t \bar{n}_w)|} \frac{dt}{\varepsilon}
\end{align*}
we then get
\begin{equation*}
U^{v^{(\varepsilon)}}(z) - U^{v}(z)
= \int v(w) \, s(dw) \, \int_0^\varepsilon \log \frac{1}{|1 + \frac{t \bar{n}_w}{z-w}|} \frac{dt}{\varepsilon}
= O(\sqrt{\varepsilon}) \int \frac{v(w) \, s(dw)}{\sqrt{|z-w|}}
= O(\sqrt{\varepsilon} M^{A_v})
,
\end{equation*}
where in the second inequality we used
$
\int_0^\varepsilon \log |1 + t\xi| \frac{dt}{\varepsilon}
\leq C\sqrt{\varepsilon |\xi|}
$,
and in the last inequality we used assumption \eqref{e:Ai}.
The potential energy is estimated using assumption (iii) by
\begin{align*}
  |(W,\mu_W^{(\varepsilon)}) - (W,\mu_W)|
  &\leq \int v(z) \, s(dz) \, \int_0^\varepsilon |W(z-t\bar n_z)-W(z)| \, \frac{dt}{\varepsilon}
  \\
  &\leq \varepsilon \|\nabla W\|_{\infty,S_W \cap F^c} + \sup_{t\in[0,\varepsilon]} \int_{F \cap \partial\Sigma_W} |W(z-t\bar n_z)-W(z)| \, v(z) \, s(dz)
  \\
  &\leq \varepsilon M^{A_\nabla} + CM^{-1}
  ,
\end{align*}
where in the last step we also used $\varepsilon < M^{-A_n}$.
To sum up, we have
\begin{equation}
\label{e:ub-energy}
| I_W(\mu_W^{(\varepsilon)}) - I_W(\mu_W)|
\leq C \sqrt{\varepsilon} M^{A_v} +  C \varepsilon M^{A_\nabla} + CM^{-1},
\end{equation}
and from \eqref{e:ub-entropy} and \eqref{e:ub-energy} it is then clear that taking $\varepsilon \leq M^{-1-\max(2A_v,A_\nabla,A_n)}$ gives \eqref{e:Z-ub}.
\end{proof}

\subsection{Proof of Proposition~\ref{prop:step}}

\begin{proof}[Proof of Proposition~\ref{prop:step}]
By Lemmas~\ref{lem:ind-lb}--\ref{lem:ind-ub}, for any $f \in C^2(\C)$
that satisfies the assumptions of Lemma~\ref{lem:ind-lb}, we have
\begin{align}
  \E_{M,W,\beta}(\ee^{\sum_j f(z_j) - M \int f \, d\mu_W})
  = \frac{Z(f)}{Z(0)} \ee^{- M \int f \, d\mu_W}
  &\leq \ee^{\frac{1}{8 \pi \beta} (f,-\Delta f) +  C \beta M \log M + A_\zeta M \log M}
  \nonumber \\
  \label{e:momest}
  &\leq \ee^{ \frac{1}{8 \pi \beta} (f,-\Delta f) + C (1+\beta) M \log M},
\end{align}
where $\E_{M,W,\beta}$ is the expectation associated to $P_{M,W,\beta}$.
Given $f$ as in the assumption of Proposition~\ref{prop:step}, we apply the above estimate with $f$ replaced by $g = s^{-1} f$, where
\begin{equation*}
  s = \frac{1}{\alpha \beta M} \|\Delta f\|_\infty + (\beta^2 M \log M)^{-\frac12} \|\nabla f\|_2.
\end{equation*}
Clearly, $|\Delta g| \leq \alpha \beta  M \leq \beta M \Delta W$ in $S_W$, so $g$ satisfies the assumption of Lemma~\ref{lem:ind-lb}. Moreover,
\begin{equation*}
\frac{1}{\beta} (g,-\Delta g)
= \frac{1}{\beta} s^{-2}(f,-\Delta f)
= \frac{1}{\beta} s^{-2}\|\nabla f\|_2^2
\leq \beta M \log M,
\end{equation*}
where in the second equality we used \eqref{e:Deltafint}.
By Markov's inequality, for any $\xi \geq 1+1/\beta$, this gives
\begin{align*}
  &P_{M,W,\beta}\left(\left|\sum_j f(z_j) -M\int f \, d\mu_W\right| \geq 3 C \xi \beta s M \log M\right)
  \\
  &\qquad \leq \ee^{ \frac{1}{8 \pi \beta s^2} (f,-\Delta f) + C (1+\beta) M\log M - 3 C \xi\beta M \log M }
  \leq \ee^{-\xi \beta M \log M}.
\end{align*}
Together with the definition of $s$, we obtain that with probability at least $1-\ee^{-\xi \beta M \log M}$,
\begin{equation*}
  \left|\frac{1}{M} \sum_j f(z_j) - \int f \, d\mu_W\right|
  \leq 3 C \xi \left(
    \frac{M \log M}{\alpha M^2} \|\Delta f\|_\infty +
    \left(\frac{M \log M}{M^2} \right)^{1/2} \|\nabla f\|_2
  \right)
  ,
\end{equation*}
as claimed. This completes the proof of Proposition~\ref{prop:step}.
\end{proof}

\section{Conditional measure}
\label{sec:condmeas}

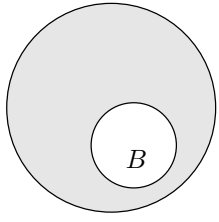
\begin{figure}
\begin{center}
\input{B0.pspdftex}
\end{center}
\caption{We condition on the particles outside $B$. \label{fig:B0}}
\end{figure}

\subsection{Conditional measure}

Let $B \subset S_V$ be compact.
For $z \in \C^N$,
we denote by $M = M(z)$ the number of particles in $z$ inside $B$,
by $\tilde z = (\tilde z_1, \dots, \tilde z_M)$ the particles contained in $B$
(ordered in an arbitrary way),
and by $\hat z = (\hat z_1, \dots, \hat z_{N-M})$ the particles outside $B$
(again ordered arbitrarily).
Then
\begin{align*}
  H_{N,V}(z)
  &= \sum_{j \neq k} \log \frac{1}{|\tilde z_j-\tilde z_k|}
  + N \sum_{j} \pB{ V(\tilde z_j) - V_o(\tilde z_j| \hat z) }
  + E(\hat z),
\end{align*}
with
\begin{equation*}
  V_o(w|\hat z) = -\frac{2}{N} \sum_{k} \log \frac{1}{|w-\hat z_k|},
                  \qquad
  E(\hat z) = \sum_{j \neq k}  \log \frac{1}{|\hat z_j-\hat z_k|} + N \sum_{j} V(\hat z_j)
  .
\end{equation*}
Thus $E(\hat z)$ is the contribution to the energy of the particles outside $B$
and $V_o$ the interaction energy between the inside and the outside particles.
Moreover, for $\hat z \in (\C \setminus B)^{N-M}$ and $z \in \C$, we set
\begin{align}
  \label{e:Wdef}
  W(w|\hat z) &= \begin{cases}
    \frac{N}{M} (V(w) - V_o(w|\hat z)) & (w \in B),\\
    +\infty& (w \not\in B),
  \end{cases}
  \\
  \label{e:Pconddef}
  P_{N,V,\beta}(dw | \hat z) &=  P_{M(\hat z),W( \cdot | \hat z),\beta}(d w).
\end{align}

\begin{definition}
We say that a function $F: \C^N \to \R$ is \emph{symmetric} if it is invariant
under any permutation of its arguments.
For $\tilde z = (\tilde z_1, \dots, \tilde z_M) \in \C^M$
and $\hat z = (\hat z_1, \dots, \hat z_{N-M}) \in \C^{N-M}$,
we write $\tilde z \circ \hat z$ for the vector
$\tilde z \circ \hat z = (\tilde z_1, \dots, \tilde z_M, \hat z_1, \dots, \hat z_{N-M}) \in \C^N$.
\end{definition}

\begin{proposition} \label{prop:condeqmeas-def}
For any symmetric $F : \C^N \to \R$,
\begin{equation*}
  \int F(z) P_{N,V,\beta}(dz) =
  \int F(w \circ \hat z) \, P_{N,V,\beta}(dw | \hat z) \, P_{N,V,\beta}(dz).
\end{equation*}
Thus $P_{N,V,\beta}(d\tilde z|\hat z)$ is the conditional probability
of the particles inside $B$ given those outside $B$.
\end{proposition}

\begin{proof}
  By inclusion-exclusion, for any symmetric function $F: \C^N \to \R$,
  \begin{multline} \label{e:condmeas-incexl}
      \int_{\C^N} F(z) \ee^{-\beta H_N(z)} \, m^{\otimes N}(dz)
      \\
      = \sum_{M=0}^N \binom{N}{M}  \int_{(B^c)^{N-M}} \int_{B^M} F(\tilde z \circ \hat z) \ee^{-\beta H_N(\tilde z \circ \hat z)} \, m^{\otimes M}(d\tilde z) \, m^{\otimes (N-M)}(d\hat z).
  \end{multline}
  Define
  \begin{equation*}
    U(z) =
    \begin{cases}
      \frac{N}{N-M} V(z) & (z \not\in B),\\
      +\infty& (z \in B).
    \end{cases}  
  \end{equation*}
  For $\tilde z \in B^M$ and $\hat z \in (B^c)^{N-M}$, we write
  $H(\tilde z|\hat z) = H_{M, W(\cdot|\hat z)}(\tilde z)$
  and $\hat H(\hat z) = H_{N-M, U}(\hat z)$. Then
  \begin{equation*}
    H_{N}(\tilde z \circ \hat z) = 
    H(\tilde z|\hat z) + \hat H(\hat z),
  \end{equation*}
  and therefore \eqref{e:condmeas-incexl} equals
  \begin{equation*}
    \sum_{M=0}^N \binom{N}{M}  \int_{(B^c)^{N-M}} \left( \int_{B^M} F(z) \ee^{-\beta H(\tilde z | \hat z)}  \, m^{\otimes M}(d\tilde z) \right) \ee^{- \beta \hat H(\hat z)} \, m^{\otimes (N-M)}(d\hat z).
  \end{equation*}
  Now define $\bar F : \C^N \to \R$ by
  \begin{equation*}
    \bar F(z)
    = \int F(w \circ \hat z) \, P_{N,V,\beta}(dw|\hat z),
  \end{equation*}
  where as previously $\hat z$ are the particles in $z$ outside $B$.
  Note that $\bar F$ is again symmetric and that $\bar F(z) = \bar F(\zeta \circ \hat z)$ for any $\zeta \in \C^{M(z)}$. Hence
  \begin{align*}
    \int_{B^M} F(w \circ \hat z) \ee^{-\beta H(w | \hat z)}  \, m^{\otimes M}(dw) 
    & =
    \bar F(z) \int_{B^M} \ee^{-\beta H(w | \hat z)}  \, m^{\otimes M}(dw) 
      \\
    & =
     \int_{B^M} \bar F(w \circ \hat z) \ee^{-\beta H(w | \hat z)}  \, m^{\otimes M}(dw).
  \end{align*}
  By \eqref{e:condmeas-incexl} with $\bar F$ instead of $F$, it follows that
  \begin{equation*}
    \int_{\C^N} F(z) \ee^{-\beta H_N(z)} \,m^{\otimes N}(dz)
    =
    \int_{\C^N} \bar F(z) \ee^{-\beta H_N(z)} \, m^{\otimes N}(dz)
  \end{equation*}
  as claimed.
\end{proof}

Proposition~\ref{prop:step} implies the following estimate for the conditional measure.

\begin{proposition} \label{prop:step-cond}
Assume that $\Delta V \geq \alpha_0$ on $\supp \mu_V$ for an absolute constant $\alpha_0>0$.
Fix $B \subset S_V$ compact and
$\hat z \in (B^c)^{N-M}$ such that $W = W(\,\cdot\, |\hat z)$ as defined by \eqref{e:Wdef} satisfies
assumptions (ii--iv) of Proposition~\ref{prop:step}.
Then, for any bounded $f \in C^2(\C)$ with $\supp \Delta f \subset \mathrm{int}(\supp \mu_W)$,
for $\xi \geq 1+1/\beta$, we have
\begin{equation} \label{e:step-cond}
  \left|\frac{1}{N} \sum_j f(z_j) - \int f \, d\mu_V\right|
  = O(\xi) \pa{ \frac{M \log M}{\alpha_0 N^2} \|\Delta f\|_\infty + \pa{\frac{M \log M}{N^2}}^{1/2} \|\nabla f\|_2 }
  ,
\end{equation}
with probability at least $1-\ee^{-\xi \beta M \log M}$ under $P_{N,V,\beta}(\cdot |\hat z)$.
\end{proposition}

\begin{proof}
  Since $W=+\infty$ outside $B$, we have $S_W \subset B$.
  Since $V_o$ is harmonic in $B$, it follows that $\Delta W = (N/M) \Delta V$ on $S_W$.
  In particular, $\Delta V \geq \alpha_0$ on $\supp \mu_V$ implies that
  assumption~(i) of Proposition~\ref{prop:step}
  holds with $\alpha = (N/M)\alpha_0$,
  and in the interior of the support of $\mu_W$, we have
  \begin{equation*}
    \mu_W(dw) = \frac{N}{M} \mu_V(dw)
    .
  \end{equation*}
  Since the left-hand side of \eqref{e:step-cond} is normalized
  by $N$ instead of $M$, and since
  \begin{equation*}
    \frac{M}{N} \frac{M \log M}{\alpha M^2} = \frac{M \log M}{\alpha_0 N^2}, \qquad
    \frac{M}{N} \pa{\frac{M\log M}{M^2}}^{1/2}
    =
    \pa{\frac{M\log M}{N^2}}^{1/2} ,
  \end{equation*}
  the claim follows from Proposition~\ref{prop:step}.
\end{proof}

\begin{remark}
  Assume $M = N^{1-2t}$ for some $0\leq t < \frac12$.
  Then Proposition~\ref{prop:step-cond} is an estimate at scales $N^{-s}$, $s<\frac14 + \frac12 t$.
  Indeed, with $f=f_s$ as in Remark~\ref{rem:step}, the error bounds are
  \begin{equation*}
    N^{-1-2t} N^{4s}, \quad  (N^{-1-2t} N^{4s})^{\frac12},
  \end{equation*}
  which go to $0$ if $s < \frac14 + \frac12 t$.
\end{remark}

\subsection{Approximate potential}
\label{sec:Wdef}

As a preliminary step to the inductive verification of the assumptions of Proposition~\ref{prop:step-cond},
we now write the potential $W$ defined in \eqref{e:Wdef} in a suitable form.

Fix any $\psi \in C_c^\infty$ with $\psi = 0$ on $B$.
(In Section~\ref{sec:R}, $\psi$ will be chosen as a bump function increasing from $0$
to $1$ in a small annulus outside $B$.)
Then
\begin{equation*}
  V_o(z|\hat z) =
  -\frac{2}{N} \sum_{k} \ind{B^c}(\hat z_k) (1-\psi(\hat z_k)) \log \frac{1}{|z-\hat z_k|}
  -\frac{2}{N} \sum_{k} \psi(\hat z_k) \log \frac{1}{|z-\hat z_k|}.
\end{equation*}
Define
\begin{align} \label{e:nudef}
  \nu(dz | \hat z) &= \frac{1}{M} \sum_k \ind{B^c}(\hat z_k) (1-\psi(\hat z_k)) \delta_{\hat z_k}(dz)
  ,
  \\
  \label{e:Rdef}
  R(z|\hat z) &= 
  \frac{1}{M} \sum_{k} \psi(\hat z_k) \log \frac{1}{|z-\hat z_k|}
  - \frac{N}{M} \int \psi(w) \log \frac{1}{|z-w|} \, \mu_V(dw)
  ,
  \\
  \label{e:Rprimedef}
  R'(z) &=  \frac{N}{M} \int \ind{B^c}(w)(1-\psi(w)) \log \frac{1}{|z-w|} \, \mu_V(dw)
  .
\end{align}
{Thus, with $\psi$ chosen as in Section~\ref{sec:R},
$U^\nu(\cdot |\hat z)$ is essentially the potential of the charges $\hat z$ near $B$,
$R(\cdot)$ is the potential of the equilibrium measure $\mu_V$ near $B$,
and $R(\cdot |\hat z)$ is the potential generated by the charges and the equilibrium measure far away from $B$.
Then
\begin{equation} \label{e:Vorewritten}
  V_o(z | \hat z) = - 2\int_{B^c} \log\frac{1}{|z-w|} \,\mu_V(dw) - \frac{M}{N} \pB{ 2U^{\nu}(z|\hat z)  + 2R(z|\hat z) - 2R'(z)},
\end{equation}
and therefore, with the definition \eqref{e:Wdef},
\begin{equation} \label{e:Wredef}
  W(z|\hat z) = \begin{cases}
    \frac{N}{M} \left(V(z) + 2 \int_{B^c} \log\frac{1}{|z-w|} \,\mu_V(dw) \right) + 2U^{\nu}(z|\hat z) + 2R(z|\hat z) - 2 R'(z) & (z \in B)
    \\
    +\infty & (z \not\in B).
  \end{cases}
\end{equation}
Finally, we set
\begin{equation} \label{e:taudef}
  \tau = \tau(\hat z) =  \frac{N}{M} \mu_V(B) .
\end{equation}
The main contribution to $W(z|\hat z)$ will be given by
\begin{equation} \label{e:barWdef}
  \bar W (z) = \begin{cases}
    \frac{1}{\mu_V(B)}  \left(V(z) + 2 \int_{B^c} \log\frac{1}{|z-w|} \,\mu_V(dw) \right) & (z \in B)\\
    +\infty & (z \not\in B).
  \end{cases}
\end{equation}
By Proposition~\ref{prop:condeqmeas}, the equilibrium measure of the potential $\bar W$
is explicitly given by a rescaling of the restriction of $\mu_V$ to $B$.
Our final goal is to show that the contribution of $W(z|\hat z)-\tau \bar W(z)$ can
be controlled, for most $\hat z$ distributed under the original measure.

\section{Proof of Theorem~\ref{thm:locallaw}}
\label{sec:pf}

Fix $\delta \in (0,\frac12)$. We will show that Theorem~\ref{thm:locallaw} holds for $s \in (0,\frac12-\delta]$,
with the implicit constant depending on $s$ only through the choice of $\delta$.
Our strategy is to apply Proposition~\ref{prop:step-cond} inductively, with $M$ approximately given by $\bar M_j=N^{1-2s_j}$
for a deterministic sequence of scales $0=s_0 < s_1 < \cdots < s_n = \frac12 - \delta$.
More precisely, our goal is to show that if the conclusions of Proposition~\ref{prop:step-cond}
hold for $B = B(z_0,\frac12 N^{-s_j})$ for some $j$, then we can verify them for $j+1$ with
\begin{equation} \label{e:sjplus1}
  s_{j+1} = \frac14 + \frac12 s_j - \varepsilon
  .
\end{equation}
The solution to this recursion is
\begin{equation} \label{e:sjsln}
  s_j = \pa{\frac14-\varepsilon} \sum_{k=0}^{j-1} \frac{1}{2^k}
  .
\end{equation}
We may choose $n < \infty$ and $\varepsilon > 0$ so that $s_n = \frac12 - \delta$.
Throughout the following, we assume that $n$ and $\varepsilon>0$ have been fixed this way.

From now on, without loss of generality, we assume $z_0=0$ in the statement of Theorem~\ref{thm:locallaw},
by replacing $V(z)$ by $V(z-z_0)$. 
In particular, then $0 \in \mathrm{int}(\supp \mu_V)$.
Moreover, we write $B_s$ for the disk centered at $0$ with radius $N^{-s}$ and $B_s^\circ \subset B_s$ for the disk
centered at $0$ with radius $\frac12 N^{-s}$.

The following proposition is the core of the proof of Theorem~\ref{thm:locallaw}.
In this statement and throughout the remainder of this section,
we say that an event holds \emph{with $t$-high probability}, abbreviated \emph{$t$-HP},
if it holds with probability at least $1-\ee^{-(1+\beta)N^{1-2t}+O(\log N)}$,
and we will use tacitly that intersections of $N^{O(1)}$ many events that each hold with $t$-HP
also hold with $t$-HP.

\medskip
\noindent \textbf{Assumption (A$_{\mathbf{t}}$).}
For any bounded $f \in C^2(\C)$ with $\supp \Delta f \subset B_t^\circ \cap S_V$,
with $t$-HP,
we have
\begin{equation} \label{e:step-prev}
  \left|\frac{1}{N} \sum_j f(z_j) - \int f \, d\mu_V\right|
  = O\pbb{\pbb{1+\frac{1}{\beta}}\log N} \pbb{ N^{-1-2t} \|\Delta f\|_\infty + N^{-\frac12-t} \|\nabla f\|_2
  }
  .
\end{equation}

\begin{proposition} \label{prop:ind}
For arbitrary $\varepsilon>0$, (A$_t$) implies (A$_s$) for any $0 {\cob \leq} t \leq s \leq \frac14 + \frac12 t - \varepsilon$,
with the implicit constant in \eqref{e:step-prev} depending only on $\varepsilon$.
\end{proposition}

The proof of the Proposition~\ref{prop:ind} is given in Section~\ref{sec:R}.
Assuming the proposition, the proof of Theorem~\ref{thm:locallaw} is completed easily,
as follows.

\begin{proof}[Proof of Theorem~\ref{thm:locallaw}]
As discussed previously, without loss of generality we assume that $z_0=0$.
As discussed around \eqref{e:sjplus1}--\eqref{e:sjsln},
given $\delta>0$ as above, we choose $\varepsilon>0$ sufficiently small that there exists $n<\infty$
such that $s_n = \frac12 -\delta$, with the sequence $(s_j)$ defined by $s_0=0$ and \eqref{e:sjplus1}.

We will apply Proposition~\ref{prop:ind} inductively.
For this, we first verify (A$_t$) for $t={\cob s_0}$.
This follows from Proposition~\ref{prop:step} applied to $W=V$ and $M=N$, for which we verify the assumptions now.
Assumptions~(i)--(iii) of Proposition~\ref{prop:step} follow directly from the assumptions of the theorem,
with $A_u = A_v= A_\nabla = A_n = 0$ and $F=\emptyset$, using that $v=0$ since $V$ is $C^{1,1}$ on a neighbourhood of $S_V$.
Therefore it suffices to verify that (iv) holds.
By Theorem~\ref{thm:eqmeasure}, or more precisely by the Euler--Lagrange equations \eqref{e:EL} and since $\mu_V$ has compact support,
we have
\begin{equation*}
  \zeta \geq 0,\qquad
  \zeta(z) \sim 2\log \frac{1}{|z|} + V(z), \quad |z|\to\infty.
\end{equation*}
By assumption \eqref{e:Vgrowth},
we thus have $\zeta(z) \geq \varepsilon \log |z|$ for sufficiently large $|z|$,
and therefore with $N$ large enough (depending on $\beta$ and $\varepsilon$ in \eqref{e:Vgrowth})
assumption~(iv) of Proposition~\ref{prop:step} follows with $A_\zeta = 0$.
Thus the assumptions of Proposition~\ref{prop:step} are verified.
It follows that (A$_{\cob s_0}$) holds.

By induction, Proposition~\ref{prop:ind} implies that (A$_t$) holds for all $t=s_j$ for all $j \leq n$,
in particular for $j=n$, and the claim follows immediately.
\end{proof}

The remainder of this section is devoted to the proof of Proposition~\ref{prop:ind}.

\subsection{Proof of Proposition~\ref{prop:ind}}
\label{sec:R}

To prove Proposition~\ref{prop:ind}, we first establish a sequence of lemmas.
We abbreviate $a=N^{-c\varepsilon}$ for a small constant $c>0$ chosen in the proofs of the lemmas,
and we fix $\psi \in C_c^\infty$ such that $\psi(z) = 0$ for $\dist(z,B) \leq aN^{-s}$
and $\psi(z) = 1$ for $\dist(z,B) \geq 2aN^{-s}$.
The following three lemmas apply to the definitions of $\tau$, $\nu(dz) = \nu(dz|\hat z)$, $\hat R(z) = R(z|\hat z)$, and $R'(z)$
from \eqref{e:taudef}, \eqref{e:nudef}, \eqref{e:Rdef}, and \eqref{e:Rprimedef}, respectively,
with this choice of $\psi$.

Sometimes we will omit the scale index $s$ or $t$, in which case it is implicitly understood to be $s$; for example,
we abbreviate $B=B_s$ and $M = M_s$.

\begin{lemma} \label{lem:taunubd}
Assume (A$_t$). Then, with $t$-HP, we have
\begin{equation} \label{e:Mbd}
\tau = 1+O(N^{-c\varepsilon}),
\qquad
\nu(\C) = O(N^{-c\varepsilon}).
\end{equation}
\end{lemma}

\begin{proof}
  We show the first bound of \eqref{e:Mbd}; the second bound is analogous.
  Fix $\eta>0$ and define
  \begin{equation*}
    B_{-} = \{ \dist(z, B^c) \geq \eta N^{-s} \}, \qquad
    B_{+} = \{ \dist(z, B) \leq \eta N^{-s} \}.
  \end{equation*}
  Let $\chi_\pm$ be smooth cutoff functions with
  \begin{equation*}
    \chi_+|_{B} = 1, \quad \chi_+|_{B_+^c} = 0, \qquad
    \chi_-|_{B^c} = 0, \quad \chi_-|_{B_-} = 1,
  \end{equation*}
  obeying $\|\nabla^k \chi_\pm\|_\infty = O(N^{k s}/\eta^k)$ for $k=0,1,2$.
  In particular, we then have
  \begin{equation*}
    \|\Delta \chi_\pm\|_\infty = O(N^{2s}/\eta^2), \quad (\chi_{\pm},-\Delta \chi_{\pm}) = O(\eta N^{-2s} N^{2s}/\eta^2) = O(1/\eta),
  \end{equation*}
  where in the second bound we used that $\Delta \chi_\pm = 0$ except on an annulus of area $O(\eta N^{-2s})$.
  Using that $s \leq \frac14 + \frac12 t -\varepsilon$, it follows that $-1-2t \leq -4s -4\varepsilon$ and thus
  \begin{equation*}
    N^{-1-2t} \|\Delta \chi_\pm\|_\infty = O(N^{-2s-4\varepsilon}\eta^{-2}), \quad
    N^{-1-2t} (\chi_\pm,-\Delta\chi_\pm) = O(N^{-4s-4\varepsilon}\eta^{-1}).
  \end{equation*}
  By \eqref{e:step-prev} therefore, with $t$-HP,
  \begin{align*}
    \frac{1}{N} \sum_j \chi_\pm(z_j)
    &= \int \chi_\pm \, d\mu_V + O(N^{-2s})(N^{-3\varepsilon}\eta^{-2}+N^{-\varepsilon}\eta^{-1/2})
    \nonumber\\
    &= \mu_V(B)\p{ 1+O(\eta + N^{-3\varepsilon}\eta^{-2} + N^{-\varepsilon}\eta^{-1/2})}
    = \mu_V(B)\p{1+ O(N^{-c\varepsilon})}
    ,
  \end{align*}
  where in the second equality we used
  \begin{equation*}
    \int \chi_\pm \, d\mu_V = \mu_V(B) + O(\eta N^{-2s}), \quad \mu(B) \asymp N^{-2s},
  \end{equation*}
  and in the last inequality we set $\eta = N^{-\varepsilon}$.

  Since $\sum_j \chi_-(z_j) \leq M \leq \sum_j \chi_+(z_j)$ we thus have $M = N \mu_V(B)(1+O(N^{-c\varepsilon}))$,
  which shows the first bound of \eqref{e:Mbd}.
  An analogous argument shows the second bound.
\end{proof}

\begin{lemma} \label{lem:Rhatbd}
Assume (A$_t$). Then, with $t$-HP, we have
\begin{equation*}
  N^{-s}\|\nabla \hat R\|_{L^\infty(B)}
  = O(N^{-c\varepsilon})
  .
\end{equation*}
\end{lemma}

\begin{proof}
  First, we show that it suffices to prove that for any fixed $z \in B$ 
  we have
  \begin{equation} \label{e:nablaRbdz}
    |\nabla \hat R(z|\hat z)| = O(N^{s-c\varepsilon}),
  \end{equation}
  with $t$-HP.
  By a union bound, this bound then indeed holds for all $z \in B \cap N^{-3}\Z^2$ with $t$-HP.
  For $|z-w| \geq N^{-\frac12}$, 
  we have
  \begin{equation*}
    \absa{ \nabla^{2} \log \frac{1}{|z-w|} } = O(N).
  \end{equation*}
  Since $\hat z$ and the support of $\psi$ are separated by distance
  at least $N^{-\frac12}$ we obtain, using \eqref{e:Mbd}, that with $t$-HP,
  \begin{equation*}
    \|\nabla^{2} R(\,\cdot\,|\hat z)\|_{L^\infty(B)} \leq \frac{N}{M} O(N) = O(N^2).
  \end{equation*}
  It follows that, with $t$-HP,
  \begin{align*}
    \|\nabla R(\,\cdot\,|\hat z)\|_{L^\infty(B)}
    &\leq \max_{z \in B \cap N^{-4}\Z^2} |\nabla R(z|\hat z)| + N^{-3} \|\nabla^{2} R(\,\cdot\, |\hat z)\|_{L^\infty(B)}
    \nonumber\\
    &\leq O(N^{s-c\varepsilon}) + O(N^{-1})
    = O(N^{s-c\varepsilon}),
  \end{align*}
  as needed.

  For the remainder of the proof, we fix some $z \in B$, and show that \eqref{e:nablaRbdz}
  does indeed hold with $t$-HP.
  For $w \in \C$, we set $f(w) = N^{-s} \nabla(\psi(w) \log \frac{1}{|z-w|})$
  and define the annulus $A = \{ w : \dist(w,B) \in [aN^{-s},2aN^{-s}] \}$,
  with $a$ as defined previously.
  Then the following properties hold:
  \begin{equation*}
    A \subset B_t^\circ, \quad m(A)= O(aN^{-2s}), 
  \end{equation*}
  and
  \begin{equation*}  
    \Delta f=0 \text{ on $A^c$},
    \quad
    \sup_A |f| = O(|\log a|/a),
    \quad
    \sup_A |\Delta f| = O({N^{2s}}/a^3). 
  \end{equation*}
  Here and below all estimates for $f$ are component-wise (in the obvious way).
  For $s \leq \frac14 + \frac12 t -\varepsilon$, it follows that $-1-2t \leq -4s -4\varepsilon$ and thus
  \begin{align*}
    N^{-1-2t} (f,-\Delta f) &= N^{-1-2t} O(N^{-2s} N^{2s}|\log a|/a^{2}) 
    = O(N^{-4s-4\varepsilon}|\log a|/a^2)
    \\
    N^{-1-2t} \|\Delta f\|_\infty &= N^{-1-2t} O(N^{2s}/a^3) 
    = O(N^{-2s-4\varepsilon}/a^3).
  \end{align*}
  Since $f$ satisfies the assumptions of (A$_t$), with $a=N^{c\varepsilon}$ for some small $c>0$,
  we obtain from \eqref{e:step-prev} that, with $t$-HP,
  \begin{equation*}
    \frac{1}{N} \sum_j f(z_j) = \int f(w) \mu_V(dw) + O(N^{-2s-c\varepsilon}).
  \end{equation*}
  Since $N/M = O(N^{2s})$ with $t$-HP, by Lemma~\ref{lem:taunubd},
  we therefore find that the claim
  $N^{-s} \nabla \hat R(z) = \frac{1}{M} \sum_j f(z_j) - \frac{N}{M}\int f \, d\mu_V = O(N^{-c\varepsilon})$
  holds, with $t$-HP.
\end{proof}

\begin{lemma} \label{lem:Rprimebd}
  Assume (A$_t$). Then, with $t$-HP, we have
  \begin{equation*}
    N^{-s}\|\nabla R'\|_{L^\infty(B)} = O(N^{-c\varepsilon}).
  \end{equation*}
\end{lemma}

\begin{proof}
  By \eqref{e:Rprimedef} and differentiation under the integral,
  \begin{equation*}
    \nabla R'(z) = \frac{N}{M} \int \ind{B^c}(w)(1-\psi(w)) \frac{w-z}{|w-z|^2} \frac{1}{4\pi}\Delta V(w) \, m(dw)
    .
  \end{equation*}
  Let $A = z+\{w : w \in [aN^{-s},2aN^{-s}] \}$ be the 
  support of $w \mapsto \ind{B^c}(z+w)(1-\psi(z+w))$. Then,
  since $\ind{B^c}(w)(1-\psi(w)) \leq 1$,
  $\|\Delta V\|_\infty \leq \alpha_0^{-1} = O(1)$,
  and since $N/M = O(N^{2s})$ with $t$-HP, by Lemma~\ref{lem:taunubd},
  we have 
  \begin{equation*}
    N^{-s}|\nabla R'(z)| = 
    O\p{N^{s}}
    \inf_{\eta>0}\pbb{
      \int_{A \cap B_\eta(0)} \frac{m(dw)}{|w|}
      + \frac{1}{\eta} \int_{A\cap B_\eta(0)^c} \ind{B^c}(w+z)(1-\psi(w+z)) \, m(dw) 
    } .
  \end{equation*}
  Since
  \begin{equation*}
    \int_{B_\eta(0)} \frac{m(dw)}{|w|} = O(\eta), \qquad
    m(A) = O(aN^{-2s}), 
  \end{equation*}
  we obtain
  \begin{equation*}
    N^{-s}|\nabla R'(z)| 
    =
    O(N^{s})
    \inf_{\eta>0}\p{ \eta + \eta^{-1}aN^{-2s} } 
    = O(\sqrt{a}) = O(N^{-c\varepsilon}),
  \end{equation*}
  where in the second equality we choose $\eta = \sqrt{a} N^{-s}$.
\end{proof}

Using Lemmas~\ref{lem:taunubd}--\ref{lem:Rprimebd} and Proposition~\ref{prop:step-cond},
we now complete the proof of Proposition~\ref{prop:ind}.

\begin{proof}[Proof of Proposition~\ref{prop:ind}]
Fix $t$ and assume (A$_t$). 
We show that then, with $t$-HP, 
$W = W(\,\cdot\,|\hat z)$ defined in \eqref{e:Wdef},
together with $V=\bar W$ given by \eqref{e:barWdef},
obeys the assumptions of Proposition~\ref{prop:step-cond}.
To verify this, we condition on $\hat z$ such that the conclusions of Lemmas~\ref{lem:taunubd}--\ref{lem:Rprimebd} hold;
this event has $t$-HP as needed.
On this event, the number of particles in $B_s$ satisfies
$M_s = N \mu_V(B_s) \tau \in [\frac12 N \mu_V(B_s), 2 N \mu_V(B_s)]$, and we have
the estimates
\begin{equation} \label{e:nuRbd}
  \|\nu\| = O(N^{-c\varepsilon}),\qquad
  N^{-s} \|\partial_n^- (R(\cdot|\hat z) +R')\|_{\infty,\partial B} = O(N^{-c\varepsilon}),
  \qquad \tau = 1+O(N^{-c\varepsilon}).
\end{equation}

These estimates, as well as $\Delta \bar W = \mu_V(B)^{-1} \Delta V \geq c \alpha_0 N^{2s}$ in $B$,
imply that $\kappa = O(N^{-s-c\varepsilon})$ in \eqref{e:kappabd}.
Thus Proposition~\ref{prop:innerradius} implies 
\begin{equation*}
  S_W \supset B_s^\circ.
\end{equation*}
Similarly, the assumptions of Proposition~\ref{prop:bddensity} are satisfied and
it follows that
there exist $u \in L^\infty(B)$ and $v \in L^\infty(\partial B)$
such that $d\mu_W = u \, dm + v\, ds$,
where $dm$ is the Lebesgue measure on $\C$
and $ds$ the surface measure on $\partial B$, satisfying
\begin{equation*}
  u = \ind{S_W} \frac{N}{4\pi M} \Delta V, \quad N^{-s} \|v\|_{\infty,\partial B} = O(1).
\end{equation*}
From the bound on $v$ and since $B$ is disk of radius $N^{-s}$, the left-hand side of \eqref{e:Ai} is
bounded by
\begin{equation}
  \|v\|_\infty \sup_{w\in\C} \int_{\partial B} \frac{s(dz)}{\sqrt{|z-w|}} = O(N^{s/2}). 
\end{equation}
Thus assumption (ii) in Proposition~\ref{prop:step-cond} (and Proposition~\ref{prop:step}) is verified
with $A_u = 2s/(1-2s) \leq 1/(2\delta)$ and $A_v = s/(1-2s) \leq 1/(4\delta)$,
where $\delta$ is the constant fixed in Theorem~\ref{thm:locallaw}, and where we used that $M \asymp N^{1-2s}$.

Next, we verify assumption (iii). To this end, define the (random) set $F = \cup B_\eta(\hat z_j) \cap B$ with the choice $\eta=N^{-4}$
and where the union ranges over all charges $\hat z_j \not\in B$.
Then for $w \not\in F$ we have
\begin{equation*}
  |\nabla W(w | \hat z)|
  \leq \frac{N}{M} |\nabla V(w)| + \frac{2}{M} \sum_{j} \frac{1}{|w-\hat z_j|}
  = O\Big( \frac{N}{\eta M}\Big)
  = O(M^{A_\nabla}),
\end{equation*}
with $A_\nabla = 5/(1-2s) - 1 \leq 5/(2\delta)-1$, where we used $N = O(M^{1/(1-2s)})$ as above.
Moreover, 
for $t \in [0,1]$ small enough that $w-t\bar n \in B$ for all $w \in \partial B$,
\begin{align*}
  &\int_{F} |W(w -t \bar n | \hat z) - W(w|\hat z)| \, v(w) \, s(dw)
  \\
  &\qquad
  \leq \int_F \pbb{ \frac{N}{M} |V(w-t\bar n)-V(w)|
  + \frac{2}{M} \sum_{j} \absa{\log \absa{1-\frac{t\bar n}{w-\hat z_j}}} } \, v(w) \, s(dw)
  \\
  &\qquad
  = O\bigg(\frac{N}{M} \|v\|_\infty\bigg)\bigg( s(F) \, t\, \|\nabla V\|_{\infty,B} + \sup_{\dist(z,B) \leq 2} \int_F |\log |w-z|| \, s(dw) \bigg)
  \\
  &\qquad
  = O(N^{1+3s} \eta \log \eta) = O(N^{-1}) = O(M^{-1}),
\end{align*}
using that 
$\|v\|_\infty = O(N^s)$, $N/M = O(N^{2s})$,
$s(F) \leq \sum_j s(B_\eta(\hat z_j)) = O(N\eta)$,
and similarly
\begin{equation*}
\sup_{\dist(z,B)\leq 2} \int_F |\log|w-z|| \, s(dw)
\leq \sum_{j} \sup_{\dist(z,B) \leq 2} \int_{B_\eta(\hat z_j)} |\log|w-z|| \, s(dw) = O(N\eta\log\eta).
\end{equation*}
Moreover, since $B$ is a disk of radius $N^{-s} \geq cM^{-s/(1-2s)} \geq cM^{-1/(4\delta)}$,
the condition assumed on $t$ is satisfied for $t \in [0,M^{-A_n}]$ with  $A_n = 1/(4\delta)+1$.

To verify assumption (iv), we recall that $\zeta \geq 0$ and $\zeta = +\infty$ outside $B$, and therefore
\begin{equation*}
  \int \ee^{-\beta N\zeta} \, dm \leq m(B) \leq 1,
\end{equation*}
so the first condition holds with $A_\zeta = 0$,
and the second condition holds with $A_n$ as previously.

Thus we have shown that $W(\,\cdot\,|\hat z)$ indeed obeys the assumptions of Proposition~\ref{prop:step-cond}
for a set of $\hat z$ that has $t$-HP; we denote the event of such $\hat z$ by $\Omega$.
For any such $\Omega$, Proposition~\ref{prop:condeqmeas-def} applied with an indicator function
implies
\begin{multline} \label{e:Pgoodbad}
  P_{N,V,\beta}\left(\left|\frac{1}{N} \sum_j f(z_j) - \int f(z) \, \mu_V(dz)\right| \geq \kappa \right)
  \leq
  P_{N,V,\beta}(\Omega^c) 
  \\
  + P_{N,V,\beta}\left(\left|\frac{1}{N} \sum_j f(z_j) - \int f(z) \,
      \mu_V(dz)\right| \geq \kappa \Bigg| \Omega \right)
  .
\end{multline}
Since $\Omega$ has $t$-HP, the first probability on the right-hand side is at most
\begin{equation*}
\ee^{-N^{1-2t} + O(\log N)} \leq \ee^{-2N^{1-2s}}.
\end{equation*}
By Proposition~\ref{prop:step-cond} with 
$\xi = 1+1/\beta$,
for $s$ as in the statement of the proposition,
the second probability in \eqref{e:Pgoodbad} is at most
$\ee^{-\xi \beta M \log M} \leq \ee^{-2(1+\beta) N^{1-2s}}$.
Since
\begin{equation*}
\ee^{-2(1+\beta) N^{1-2s}}
+ \ee^{-2(1+\beta) N^{1-2s}}
\leq \ee^{-(1+\beta) N^{1-2s}}
\end{equation*}
we conclude that (A$_s$) holds. This completes the proof.
\end{proof}

\section{Proof of Theorem~\ref{thm:rigidity}}
\label{sec:rigidity}

In this section we generally assume that the potential $V$ satisfies the assumptions of Theorem~\ref{thm:rigidity},
namely that $V \in C^4$ and that it has the growth \eqref{e:Vgrowth} at infinity.
Moreover, as in the statement of Theorems~\ref{thm:rigidity},
choose $s \in (0,\frac12)$, $z_0$ in the interior of $S_V$,
and suppose that $f \in C^4_c$ is supported in $B(z_0, \frac12 N^{-s})$.

Let
\begin{equation*}
  F_{N,V,\beta}(f)
  = \log \E_{N,V,\beta} (\ee^{X_f}),
  \qquad
  X_f = \sum_j f(z_j) - N \int f \, d\mu_V
\end{equation*}
denote the cumulant generating function $F_{N,V,\beta}$ of the centered linear statistic $X_f$
associated to $f$, where as previously $\E_{N,V,\beta}$ is the expectation associated to $P_{N,V,\beta}$.
Further define
\begin{equation*}
  \|f\|_{k,t} =  \sum_{l=1}^k t^l \|\nabla^l f\|_\infty.
\end{equation*}
Our goal is to prove that,
for any $\varepsilon>0$, for all $f \in C^4_c(\C)$ with support contained in $B(z_0, \frac12 N^{-s})$ and
$\|f\|_{4,N^{-s}}\leq \beta \alpha_0$, where $\alpha_0$ is fixed in the assumptions of Theorems~\ref{thm:locallaw}--\ref{thm:rigidity},
we have
\begin{equation} \label{e:FNbd}
F_{N,V,\beta}(f)
=
O(\beta N^{\varepsilon})
.
\end{equation}
For this, it suffices to show that 
$\ddp{}{t} F_{N,V,\beta}(tf)$ is bounded by the right-hand side of \eqref{e:FNbd}
uniformly in $t \in [0,1]$.
Since $F_{N,V,\beta}(0)=0$ the claim will then follow by integration over $t\in [0,1]$.

\subsection{Loop equation}

To bound the derivative $\ddp{}{t} F_{N,V,\beta}(tf)$ we use the \emph{loop equation} and estimate the arising
error terms using Theorem~\ref{thm:locallaw}.
For the following, we recall the notation
$\partial = \frac12 (\partial_x - \ii \partial_y)$ and $\bar\partial = \frac12 (\partial_x + \ii \partial_y)$,
and that $\partial\bar\partial  = \frac{1}{4} \Delta$.

\begin{proposition} \label{prop:loop}
For any potential $V \in C^1$ with sufficient growth at infinity,
and any test function $h\in C^1_c$, we have the \emph{loop equation}
\begin{equation} \label{e:loop}
  \E_{N,V,\beta} \pa{
    \half \sum_{j\neq k} \frac{h(z_j)-h(z_k)}{z_j-z_k} +
    \frac{1}{\beta} \sum_j \partial h(z_j)
    - N \sum_j h(z_j) \partial V(z_j)
  } = 0.
\end{equation}
\end{proposition}

\begin{proof}
By integration by parts, we have
\begin{equation*}
  \E_{N,V,\beta}\pa{\partial h(z_j)}
  = \beta \E_{N,V,\beta} \p{ h(z_j) \partial_{z_j} H(z) }
  = \beta \E_{N,V,\beta} \pa{ h(z_j) \pa{ \sum_{k: k \neq j}\frac{-1}{z_j-z_k} + N\partial V(z_j)} }
  ,
\end{equation*}
and \eqref{e:loop} follows immediately from this equation by summation over $j$.
\end{proof}

The following lemma rewrites the Euler--Lagrange equation \eqref{e:EL}
in a form useful for the reformulation of the loop equation.

\begin{proposition} \label{prop:ELff}
For any sufficiently smooth $f$ supported in the interior of $S_V$, we have
\begin{equation} \label{e:ELff}
  \half \iint \frac{f(z)-f(w)}{z-w} \mu_V(dz) \, \mu_V(dw) - \int f(z) \partial V(z) \mu_V(dz) = 0
\end{equation}
and $K_V(\frac{\bar\partial f}{\Delta V}) = \frac{1}{4} f$ where
\begin{equation}
  K_Vf(z)
  %= \int \partial g(z-w)((f(z)-f(w)) \; \mu_V(dw) + \partial V(z) f(z).
  = -\int \frac{f(z)-f(w)}{z-w} \; \mu_V(dw) + \partial V(z) f(z).
\end{equation}
\end{proposition}

\begin{proof}
By the Euler-Lagrange equation \eqref{e:EL},
$U^{\mu_V} + \frac12 V$ is constant in the support of $\mu_V$.
Thus, using $\partial \log 1/|z| = -1/(2z)$,
for $z$ in the support of $\mu_V$ we have
\begin{equation}
  - \int \frac{1}{z-w} \; \mu_V(dw) = 2 \partial U^{\mu_V}(z) = - \partial V(z).
\end{equation}
This implies \eqref{e:ELff}, and since $\mu_V = \frac{1}{4\pi} \Delta V$ on its support,
also by \eqref{e:EL}, 
and since $\partial\bar\partial = \frac{1}{4} \Delta$, also
\begin{equation}
  K_V \pa{ \frac{\bar\partial f}{\Delta V}}(z)
  = \frac{1}{4\pi} \int \frac{\bar\partial f(w)}{z-w} \, m(dw)
  = \frac{1}{2\pi} \int  \partial \bar\partial f(w) \log |z-w| \, m(dw)
  = \frac{1}{4} f(z),
\end{equation}
as claimed.
\end{proof}

For the following, we recall that $\hat \mu = \frac{1}{N} \sum_j \delta_{z_j}$ denotes
the empirical measure, and define $\tilde \mu_V = \hat \mu - \mu_V$ to be the difference
between the empirical and equilibrium measures. Given $f$, we moreover abbreviate
\begin{equation} \label{e:hdef}
  h = \frac{4\bar\partial f}{\Delta V}
\end{equation}
throughout the remainder of Section~\ref{sec:rigidity}.

\begin{proposition} \label{prop:loopF}
For any sufficiently smooth $f$ supported in the interior of $S_V$, we have
\begin{multline} \label{e:loopF}
  \ddp{}{t} F_{N,V,\beta}(tf)
  =
  \E_{N,V-tf/(\beta N),\beta} \bigg(
  \frac{1}{\beta} \int \partial h \, d\hat\mu
  + \frac{t}{\beta} \int h \partial f \, d\hat\mu
  \\
  + \frac{N}{2} \iint \frac{h(z)- h(w)}{z-w} \inda{z\neq w} \; \tilde\mu_V(dz) \, \tilde\mu_V(dw)
  \bigg).
\end{multline}
% with $(h(z)-h(w))/(z-w)$ interpreted as $\partial h(z)$ for $z=w$.
\end{proposition}

\begin{proof}
We use the loop equation \eqref{e:loop} with $V$ replaced by $V-tf/(\beta N)$.
The first term in \eqref{e:loop} is proportional to
\begin{align} \label{e:looppf1}
  \frac{1}{2N^2} \sum_{j\neq k} \frac{h(z_j)-h(z_k)}{z_j-z_k}
  &=
  \half \iint \frac{h(z)-h(w)}{z-w} \, \inda{z \neq w} \, \hat\mu(dz) \, \hat\mu(dw)
  \nonumber
  \\
  &=
  \half \iint \frac{h(z)-h(w)}{z-w} \, \mu_V(dz) \, \mu_V(dw)
  + \iint \frac{h(z)-h(w)}{z-w} \, \mu_V(dz) \, \tilde\mu_V(dw) 
  \nonumber 
  \\
  &\qquad
  + \half \iint \frac{h(z)-h(w)}{z-w} \, \inda{z \neq w} \, \tilde\mu_V(dz) \, \tilde\mu_V(dw)
  ,
\end{align}
where we used that $\mu_V$ is absolutely continuous to omit the indicator function $\inda{z\neq w}$
in the first two terms on the right-hand side.
Similarly, we write the last term in \eqref{e:loop} as
\begin{equation} \label{e:looppf2}
  \frac{1}{N} \sum_j h(z_j) \partial V(z_j)
  = \int h(z) \partial V(z) \, \mu_V(dz)
  + \int h(z) \partial V(z) \, \tilde\mu_V(dz).
\end{equation}
By Proposition~\ref{prop:ELff}, the difference of the first terms on the right-hand sides of \eqref{e:looppf1} and \eqref{e:looppf2}
vanishes, and the difference of the second terms is equal to
\begin{equation}
  \int K_V h(w) \, \tilde\mu_V(dw).
\end{equation}
Applying \eqref{e:loop} with $V$ replaced by $V - tf/(\beta N)$, we obtain an additional term involving $f$ and find
\begin{multline}
  \E_{N,V-tf/(\beta N),\beta} \pa{ N \int K_Vh \, d\tilde\mu_V }
  =
  \E_{N,V-tf/(\beta N),\beta}
  \bigg(
  \frac{1}{\beta} \int \partial h \, d\hat\mu
  + \frac{t}{\beta} \int h \partial f \, d\hat\mu
  \\
  +  \frac{N}{2} \iint \frac{h(z)-h(w)}{z-w} \, \inda{z\neq w} \, \tilde\mu_V(dz) \, \tilde\mu_V(dw)
  \bigg).
\end{multline}
Let $h=4\bar\partial f/\Delta V$. By Proposition~\ref{prop:ELff} we then have
$K_V h = f$ and thus $N \int K_Vh \, d\tilde \mu_V = X_f$ and obtain \eqref{e:loopF}.
\end{proof}

\subsection{Proof of Theorem~\ref{thm:rigidity}}

To show that the right-hand side of \eqref{e:loopF} is bounded by the
right-hand side of \eqref{e:FNbd},
we use Theorem~\ref{thm:locallaw} applied with potential $V-tf/(\beta N)$ instead of $V$,
where $f \in C_c^4(\C)$ has support in $B(z_0,\frac12 N^{-s})$ and $\|f\|_{4,N^{-s}} \leq \beta \alpha_0$.

We choose $\delta  \in (0,\frac12 -s)$.
In particular, then $\|\nabla f\|_\infty/(\beta N) \leq N^{-\frac12-\delta}\alpha_0$ and $\|\Delta f\|_\infty/(\beta N) \leq N^{-\delta} \alpha_0$,
so Theorem~\ref{thm:locallaw} can be applied to the potential $V-tf/(\beta N)$ with constants independent of $f$.
Moreover, by Proposition~\ref{prop:eqmeasf} and since $t |\Delta f| \leq \beta N \Delta V$ on the support of $\mu_V$,
we have $\mu_{V-tf/(\beta N)} = \mu_V - \frac{1}{4\pi \beta N}  \Delta f \, dm$ and therefore
\begin{equation} \label{e:muVmuVf}
  \int g \, d\mu_{V} - \int g \, d\mu_{V-tf/(\beta N)}
  = \frac{1}{4\pi \beta N} \int g \Delta f\, dm
  \leq CN^{-1} \|g\|_\infty
  \leq CN^{-1} \|g\|_{2,N^{-s}},
\end{equation}
where we used $m(\supp \Delta f) = O(N^{-2s})$ and $|\Delta f| \leq N^{2s}\beta\alpha_0$.
Thus we can replace $\mu_{V-tf/(\beta N)}$ by $\mu_V$ on the left-hand side of the estimate in Theorem~\ref{thm:locallaw}
with a negligible error.

In the following, we say that an event $E$ holds with \emph{$\delta$-high probability} (abbreviated \emph{$\delta$-HP})
if
\begin{equation} \label{e:deltaHP}
  P_{N,V-tf/(\beta N),\beta}(E) \leq \ee^{-(1+\beta)N^{2\delta}+O(\log N)}
  .
\end{equation}
Except for the replacement of $V$ by $V-tf/(\beta N)$,
this definition is the same as $t$-HP, introduced in the previous section above \eqref{e:step-prev}, with $t=\frac12-\delta$.
However, in this section, it is more natural to use $\delta$ rather than $\frac12-\delta$. As previously,
by the definition and the union bound, any union of $N^{O(1)}$ many events which each hold with $\delta$-HP also holds with $\delta$-HP.
Throughout the remainder of this section, we also abbreviate
\begin{equation} \label{e:thetaxidef}
\theta=N^{\delta}, \qquad \xi = (\log N)(1+1/\beta).
\end{equation}
Then, by Theorem~\ref{thm:locallaw} and \eqref{e:muVmuVf}, it follows that
for any $g \in C_c^2$ with support in a ball of radius $t$ with $t > \theta/\sqrt{N}$ we have, with $\delta$-HP,
\begin{equation} \label{e:locallawprec}
  \int g \, d\tilde\mu_V
  = O(\xi) \pa{ N^{-\frac12} (t^2 \wedge 1) \|\nabla g\|_\infty + N^{-1} (t^2\wedge 1) \|\Delta g\|_\infty }
  = O(\xi) N^{-\frac12} (t\wedge 1) \|g\|_{2,t\wedge 1}.
\end{equation}
To be precise, Theorem~\ref{thm:locallaw} as stated does not apply to the case $s=0$ corresponding to $t \geq 1$
as here $g$ is not assumed to have support contained in $S_V$.
However, it is immediate from Proposition~\ref{prop:step} that \eqref{e:locallawprec} holds also in this case
as by our smoothness assumptions the equilibrium measure has no boundary charge.

Below we further use the abbreviation \eqref{e:hdef}, and permit all constants to depend on $V$.
In particular, we use the bound $N^{-s} \|h\|_{k,N^{-s}} \leq C \|f\|_{k+1, N^{-s}}$.

We first estimate the first two terms in \eqref{e:loopF}.

\begin{lemma} \label{lem:Ehf}
For $f$ satisfying the conditions above,
with $\delta$-HP, we have
\begin{align}
  \label{e:Ehf1}
  \int \partial h \, d\hat\mu
  - \frac{1}{4\pi} \int (\Delta f) (\log \Delta V) \, dm
  &= O(\xi/\theta) \|f\|_{4,N^{-s}}
  ,
  \\
  \label{e:Ehf2}
  \int h \partial f \, d\hat\mu
  -   \frac{1}{4\pi} \int f(-\Delta f) \, dm  
  &= O(\xi/\theta) \|f\|_{3,N^{-s}}^2
  .
\end{align}
\end{lemma}

\begin{proof}
By definition of $h$, we have
\begin{equation} \label{e:phhpf}
  \partial h
  = 4 \frac{\partial\bar\partial f}{\Delta V} + 4 (\bar\partial f) \partial \pa{\frac{1}{\Delta V}}
  = \frac{\Delta f - 4 (\bar\partial f) (\partial \log \Delta V)}{\Delta V} 
  ,\qquad
  h \partial f = \frac{4(\bar\partial f)(\partial f)}{\Delta V}
  .
\end{equation}
By integration by parts, 
we have $\int \Delta f \, dm =0$ for $f \in C_c^2(\C)$ and,
since $d\mu_V = \frac{1}{4\pi} \Delta V \, dm$ on its support, therefore
\begin{align*}
  \int \partial h \, d\mu_V
  &=
  \frac{1}{4\pi} \int (\Delta f -4(\bar\partial f)(\partial \log \Delta V)) \, dm
  = \frac{1}{4\pi} \int (\Delta f) (\log \Delta V) \, dm
  ,
  \\
  \int h \partial f \, d\mu_V
  &=
  \frac{1}{\pi} \int (\bar\partial f) (\partial f) \, dm
  =
  \frac{1}{4\pi} \int f(-\Delta f) \, dm  
  .
\end{align*}

By \eqref{e:locallawprec} and since $N^{-\frac12} \leq N^{-\delta-s} = N^{-s}/\theta$,
to complete the proof of \eqref{e:Ehf1}--\eqref{e:Ehf2},
it only remains to verify that
$N^{-2s}\|\partial h\|_{2,N^{-s}}$ and $N^{-2s}\|h\partial f\|_{2,N^{-s}}$ are bounded by the right-hand sides
of \eqref{e:Ehf1} and \eqref{e:Ehf2}, respectively.
This is clear from \eqref{e:phhpf} and thus the proof is complete.
\end{proof}

More care is required to estimate the last term in \eqref{e:loopF}.
Indeed,
\begin{equation*}
  \frac{h(z)-h(w)}{z-w} = \partial h(z) + \bar\partial h(z) \frac{\bar z -\bar w}{z-w} + O(|z-w|),
\end{equation*}
and the second term on the right-hand side is not smooth on the diagonal.

\begin{lemma} \label{lem:Ehh}
For $f$ satisfying the conditions above,
with $\delta$-HP, we have
\begin{equation} \label{e:Ehh}
  N \iint \frac{h(z)- h(w)}{z-w} {\inda{z\neq w}}\; \tilde\mu_V(dz) \, \tilde\mu_V(dw)
  =
  O(\theta^2 + \xi^2 \log N) \|f\|_{3,N^{-s}}.
\end{equation}
\end{lemma}

\begin{proof}
To decompose the singularity, we use that for any compactly supported $\varphi: [0,\infty) \to \R$ we have
\begin{equation} \label{e:hdecomp}
  \frac{h(z)-h(w)}{z-w}
  = C \int_0^\infty \int_\C \varphi(|z-\zeta|/t)\varphi(|w-\zeta|/t) (\bar z-\bar w)(h(z)-h(w)) \, m(d\zeta) \, \frac{dt}{t^5}.
\end{equation}
Indeed, this formula is equivalent to
\begin{equation*}
  \frac{1}{|z|^2} = C \int_0^\infty \int_\C \varphi(|z-\zeta|/t)\varphi(|\zeta|/t) \, m(d\zeta) \, \frac{dt}{t^5},
\end{equation*}
which holds since both sides are rotationally symmetric and homogeneous of degree $-2$,
for any function $\varphi$ for which the double integral on the right-hand side is well-defined (see e.g.\ \cite{MR864658}).
For example, we assume for convenience that $\varphi$ is smooth, not identically zero, that $0\leq \varphi \leq 1$,
and that $\varphi$ is supported in $[\frac14,\frac12]$: this is enough for (\ref{e:hdecomp}) to hold in the absolutely convergent sense
with $0<C<\infty$.

Recall that $\theta = N^\delta$.
We first consider the singular contribution $t< \theta/\sqrt{N}$ to \eqref{e:hdecomp}, namely
\begin{equation*}
  H_0(z,w) =
  \inda{z\neq w} \int_0^{\theta/\sqrt{N}} \int_\C \varphi(|z-\zeta|/t)\varphi(|w-\zeta|/t) (\bar z-\bar w)(h(z)-h(w)) \, m(d\zeta) \, \frac{dt}{t^5}.
\end{equation*}
In this case, we use that since $h$ has support in $B(z_0,\frac12 N^{-s})$ we have
\begin{multline*}
  |\varphi(|z-\zeta|/t)(\varphi(|w-\zeta|/t) (\bar z-\bar w)(h(z)-h(w))|
  \\
  \leq \varphi(|z-\zeta|/t) \varphi(|w-\zeta|/t) t^2 \|\nabla h\|_\infty \inda{z,w \in B(z_0,N^{-s})}.
\end{multline*}
Thus we have the deterministic bound
\begin{equation*}
  |H_0(z,w)|
  \leq
  \inda{z,w \in B(z_0,N^{-s})}
  \|\nabla h\|_\infty \int_0^{\theta/\sqrt{N}} \int_\C \varphi(|z-\zeta|/t) \varphi(|w-\zeta|/t) \, m(d\zeta) \, \frac{dt}{t^3}.
\end{equation*}
The only nonvanishing contributions to the integral are for
$|z-w| \leq t \leq \theta/\sqrt{N}$ and thus the integral can be estimated by
\begin{equation*}
  \inda{|z-w| \leq \theta/\sqrt{N}}
  \int_0^\infty\int_\C \varphi(|z-\zeta|/t) \varphi(|w-\zeta|/t) \, m(d\zeta) \, \frac{dt}{t^3}
  = C \inda{|z-w| \leq \theta/\sqrt{N}},
\end{equation*}
where the constant on the right-hand side is the integral on the left-hand side
(which depends neither on $z$ or $w$ by the same argument that shows \eqref{e:hdecomp}).
In summary, we have shown that
\begin{equation} \label{e:H0bda}
  |H_0(z,w)| \leq  C \indb{z,w \in B(z_0,N^{-s})} \indb{|z-w| \leq \theta/\sqrt{N}} \|\nabla h\|_\infty.
\end{equation}
Now we estimate
\begin{equation*}
  \iint H_0(z,w) \, \tilde \mu(dz) \, \tilde \mu(dw)
  \leq 
  \iint |H_0(z,w)| \, (\hat \mu+\mu_V)(dz) \, (\hat \mu+\mu_V)(dw).
\end{equation*}
We cover $B(z_0,N^{-s})$ by $O(N^{1-2s}/\theta^2)$ many disks $B_j'$ of diameters $\theta/\sqrt{N}$,
and denote by $B_j$ the disk with the same center as $B_j'$ but diameter $2\theta/\sqrt{N}$.
By \eqref{e:H0bda} the support of $H_0$ is contained in the union
$\bigcup_j \{(z,w): z \in B_j' ,w \in B_j \} \subset \bigcup_j \{(z,w): z,w \in B_j \}$,
and thus
\begin{equation*}
  \iint |H_0(z,w)| \, (\hat \mu+\mu_V)(dz) \, (\hat \mu+\mu_V)(dw)
  \leq
  C\|\nabla h\|_\infty \sum_j (\hat \mu+\mu_V)(B_j)^2
  ,
\end{equation*}
where the sum ranges over the $O(N^{1-2s}/\theta^2)$ many disks.
For each disk, \eqref{e:locallawprec} implies that
\begin{equation*}
  \hat\mu(B_j)^2 \leq 2\mu_V(B_j)^2 = O(\theta^4/N^2)
\end{equation*}
with $\delta$-HP.
By a union bound, this holds simultaneously for all disks, again with $\delta$-HP.
Together with the number of disks, we have thus obtained the bound
\begin{equation} \label{e:H0bd}
  N \iint H_0(z,w) \, \tilde \mu(dz) \, \tilde \mu(dw)
  = O(\theta^2 N^{-2s}) \|\nabla h\|_\infty
  = O(\theta^2) \|f\|_{2,N^{-s}},
\end{equation}
with $\delta$-HP, as needed.

Next we estimate the remaining part of the integral \eqref{e:hdecomp}, namely $t \geq \theta/\sqrt{N}$.
We will first estimate the integrand for fixed such $t$ and any $\zeta$.
For this, we write
\begin{multline} \label{e:barzbarhzhw4}
  (\bar z -\bar w)(h(z)-h(w))
  = (\bar z - \bar\zeta) (h(z)-h(\zeta)) - (\bar w-\bar\zeta) (h(z)-h(\zeta))
  \\
  - (\bar z - \bar\zeta) (h(w)-h(\zeta)) + (\bar w-\bar\zeta) (h(w)-h(\zeta)).
\end{multline}
Since the four terms are analogous, we only consider the first one.
Set $u(z) = \varphi(|z-\zeta|/t)(\bar z - \bar\zeta)(h(z)-h(\zeta))$
and $v(w) = \varphi(|w-\zeta|/t)$. Then
\begin{gather*}
  \|\nabla u\|_\infty \leq Ct \|\nabla h\|_\infty, 
  \quad
  \|\nabla^2u\|_\infty \leq Ct \|\nabla^2 h\|_\infty + C \|\nabla h\|_\infty,
  \\
  \|\nabla v\|_\infty \leq C/t,
  \quad
  \|\nabla^2 v\|_\infty \leq C/t^2,
\end{gather*}
and $r = \diam\supp u=N^{-s} \wedge t$ if $h(\zeta)=0$
and $r=t$ if $h(\zeta) \neq 0$,
and the support of $v$ has diameter at most $t$.
Thus, for any $t > \theta/\sqrt{N}$ and any $\zeta \in\C$, by \eqref{e:locallawprec}
and since $r \leq t$ we obtain
\begin{align*}
  &\iint u(z) v(w) \, \tilde \mu_V(dz) \, \tilde \mu_V(dw)
  \\
  &\qquad
  =
  O\pB{ \xi N^{-\frac12} r^2 t \|\nabla h\|_{\infty} + \xi N^{-1} r^2 (\|\nabla h\| + t\|\nabla^2 h\|) }
  O\pB{ \xi N^{-\frac12} (t^2 \wedge 1)/t }
  \\
  &\qquad
  =
  O\pB{\xi^2 N^{-1} r^2 (t^2 \wedge 1) \|\nabla h\|_{\infty}}
  +
  O\pB{\xi^2 N^{-3/2} r^2 (t^2 \wedge 1) \|\nabla^2 h\|_{\infty}}
  ,
\end{align*}
with $\delta$-HP.

By a union bound, this estimate can be extended from fixed $t$ and $\zeta$ to all $t \in [\theta/\sqrt{N},N^2] \cap N^{-10}\Z$
and $\zeta \in \{\zeta' \in \C \cap (N^{-10}\Z)^2: |\zeta'| \leq N^2 \}$.
Since the left-hand side is Lipschitz continuous in $\zeta$ and $t$ with Lipschitz constant at most
$O(1+1/t^{5})$ this bound then actually holds uniformly in $\theta/\sqrt{N} \leq t \leq N^2$ and $|\zeta| \leq N^2$,
with $\delta$-HP.

Note that we may restrict the integral over $\zeta \in \C$ in \eqref{e:hdecomp}
to those $\zeta$ with distance at most $t$ to $\supp h$.
We divide the integral over $\zeta$ into the regions $\zeta \in \supp h$ and $0<\dist(\zeta, \supp h) \leq t$.
The contribution of the integral over $\zeta \in \supp h$ is $O(N^{-2s})$,
which together with the factor $r^2 =t^2$ gives $O(N^{-2s}t^2)$.
The contribution of the integral over $\zeta\not\in \supp h$ is $O(N^{-2s} \vee t^2)$,
which together with the factor $r^2$, which is $r^2 = N^{-2s} \wedge t^2$ in this case,
also gives $O(N^{-2s}t^2)$.
In summary, we have shown that
\begin{align} \label{e:uvtildemutildemu}
  &\iiint u(z) v(w) \, \tilde \mu_V(dz) \, \tilde \mu_V(dw) \, m(d\zeta)
  \nonumber\\
  &\qquad
  =
  O\pB{\xi^2 { N^{-2s}t^2(t^2 \wedge 1)} }
  \pB{ N^{-1} \|\nabla h\|_{\infty} + N^{-3/2} \|\nabla^2 h\|_{\infty} }
  ,
\end{align}
uniformly in $t \in [\theta/\sqrt{N},N^2]$, with $\delta$-HP.
Since 
{\begin{equation*}
  \int_{\theta/\sqrt{N}}^{N^2}
  \frac{t^2(t^2 \wedge 1)}{t^5} 
  \, dt
  = O(\log N),
 \end{equation*}}
integration of \eqref{e:uvtildemutildemu} over $t \in [\theta/\sqrt{N},N^2]$ with respect to the measure $dt/t^5$ shows
\begin{multline} \label{e:H1bd}
  N \iint \pbb{ \int_{\theta/\sqrt{N}}^{N^2} \int_\C \varphi(|z-\zeta|/t)\varphi(|w-\zeta|/t) (\bar z-\bar w)(h(z)-h(w)) \, m(d\zeta) \, \frac{dt}{t^5} }
  \, \tilde \mu_V(dz) \, \tilde \mu_V(dw)
  \\
  = O(\xi^2 \log N) (N^{-2s}\|\nabla h\|_\infty + N^{-3s} \|\nabla^2 h\|_\infty)
  = O(\xi^2 \log N) \|f\|_{3,N^{-s}},
 \end{multline}
 where we also used $N^{-1/2} \leq N^{-s}$.

Finally, for $t\geq N^2$, necessarily $h(z)=0$ or $h(w)=0$, and we obtain deterministically 
\begin{multline} \label{e:H2bd}
  N \int_{N^2}^\infty \int_\C \varphi(|z-\zeta|/t)\varphi(|w-\zeta|/t) (\bar z-\bar w)(h(z)-h(w)) \, m(d\zeta) \, \frac{dt}{t^5}
  \\
  \leq 
  N \|h\|_\infty \int_{N^2}^\infty \frac{dt}{t^2}
  \leq N^{-1}\|h\|_\infty \leq \|f\|_{1,N^{-s}}
  .
\end{multline}
The proof is complete since the left-hand side of \eqref{e:Ehh} is bounded by the sum of
\eqref{e:H0bd}, \eqref{e:H1bd}, and \eqref{e:H2bd}.
\end{proof}

\begin{proof}[Proof of Theorem~\ref{thm:rigidity}]
  By Proposition~\ref{prop:loopF} and Lemmas~\ref{lem:Ehf}--\ref{lem:Ehh}, and since
  the random variable in the expectation of the right-hand side \eqref{e:loopF} is almost surely bounded by
  $O(N(1+\frac{1}{\beta}))(\|\nabla^2 f\|_\infty + \|\nabla f\|_\infty^2)$,
  for all $f$ satisfying the conditions above \eqref{e:FNbd}, we have
  \begin{equation*}
    \ddp{}{t} F_{N,V-tf/(\beta N),\beta}(tf)
    =
    O\pbb{ \theta^2 + \xi^2 \log N }
    \|f\|_{3,N^{-s}}
    +
    O\pbb{\frac{\xi}{\theta \beta}}
    \pbb{\|f\|_{4,N^{-s}}+\|f\|_{3,N^{-s}}^2}
    .
  \end{equation*}
  By assumption $\|f\|_{4,N^{-s}} \leq \beta \alpha_0 = O(\beta)$, and thus since
  $\theta=N^{2\delta}$ and $\xi = (1+1/\beta)\log N$, where $\delta>0$ can be chosen arbitrarily small,
  both terms on the right-hand side are bounded by
  $O(\beta N^{\varepsilon})$ for sufficiently large $N$.
  Then \eqref{e:FNbd} follows by integration over $t\in [0,1]$.

  To prove \eqref{e:rigidity}, given any $f$ as in the assumption of Theorem~\ref{thm:rigidity},
  we use \eqref{e:FNbd} with $f$ replaced by $g= \beta \alpha_0 f/\|f\|_{4,N^{-s}}$.
  Then $F_{N,V,\beta}(g) = O\p{ \beta N^{\varepsilon} }$,
  and by Markov's inequality, we conclude that for a sufficiently large constant $C$ we have
  \begin{equation*}
    \P\pa{X_f \geq C \frac{N^{\varepsilon}}{\alpha_0} \|f\|_{4,N^{-s}}}
    =
    \P\pa{X_g \geq C \beta N^{\varepsilon}}
    \leq
    \ee^{-\beta N^{\varepsilon}}
    .
  \end{equation*}
  This completes the proof.
\end{proof}

\section*{Acknowledgements}

PB was partially supported by NSF grants DMS-1208859 and DMS-1513587.
HTY was partially supported by NSF grant DMS-1307444 and a Simons Investigator fellowship.
We also gratefully acknowledge the hospitality and support of the
Institute for Advanced Study in Princeton (RB, PB, HTY),
and the National Center for Theoretical Sciences
and the National Taiwan University in Taipei (RB, PB, MN, HTY),
where part of this research was carried out.
The authors' stay at the IAS was supported by NSF grant DMS-1128155.
MN was supported by the Center of Mathematical Sciences and Applications at Harvard University.

We thank Nestor Guillen for a helpful discussion.

\bibliography{all}
\bibliographystyle{plain}

\end{document}

%% file: zoom.pspdftex
\begin{picture}(0,0)%
\includegraphics{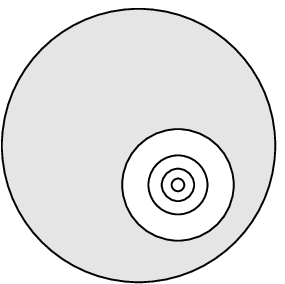}%
\end{picture}%
\setlength{\unitlength}{4144sp}%
\begingroup\makeatletter\ifx\SetFigFont\undefined%
\gdef\SetFigFont#1#2#3#4#5{%
  \reset@font\fontsize{#1}{#2pt}%
  \fontfamily{#3}\fontseries{#4}\fontshape{#5}%
  \selectfont}%
\fi\endgroup%
\begin{picture}(1266,1266)(1078,-1774)
\end{picture}%

%% file: innerradius.pspdftex
\begin{picture}(0,0)%
\includegraphics{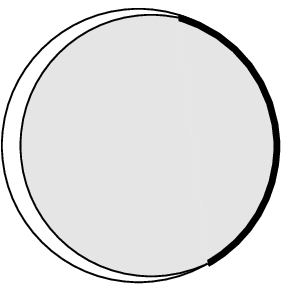}%
\end{picture}%
\setlength{\unitlength}{4144sp}%
\begingroup\makeatletter\ifx\SetFigFont\undefined%
\gdef\SetFigFont#1#2#3#4#5{%
  \reset@font\fontsize{#1}{#2pt}%
  \fontfamily{#3}\fontseries{#4}\fontshape{#5}%
  \selectfont}%
\fi\endgroup%
\begin{picture}(1289,1266)(1078,-1774)
\end{picture}%

%% file: lfit.pspdftex
\begin{picture}(0,0)%
\includegraphics{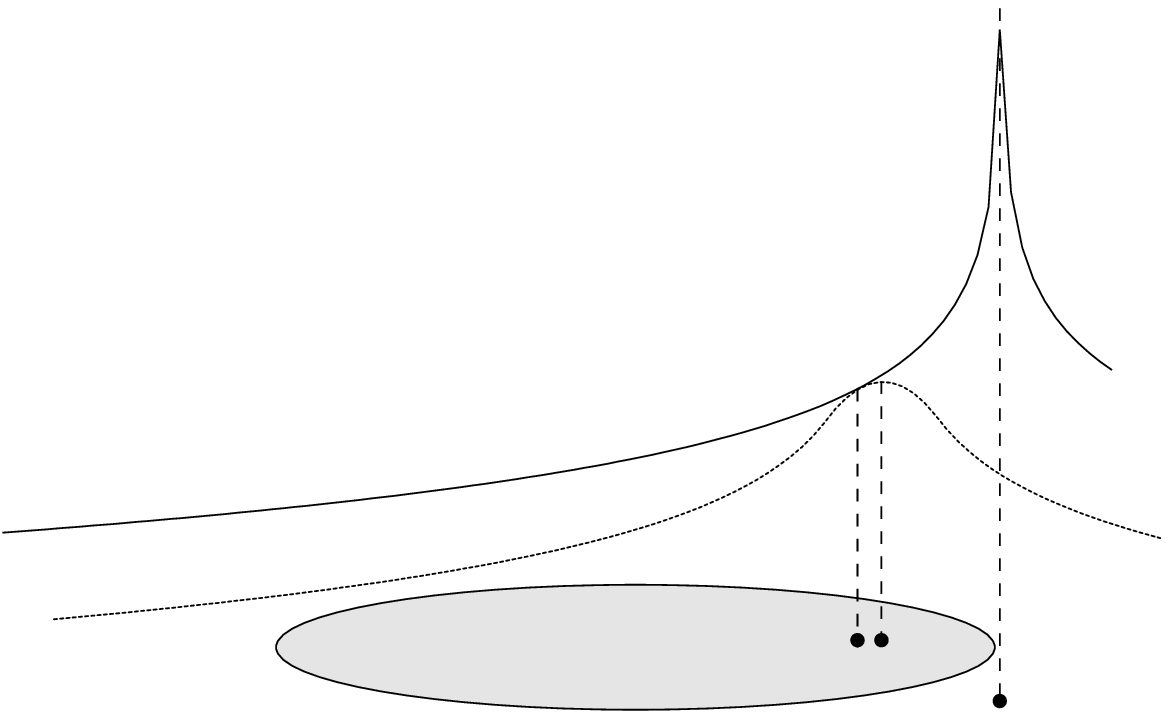}%
\end{picture}%
\setlength{\unitlength}{3947sp}%
\begingroup\makeatletter\ifx\SetFigFont\undefined%
\gdef\SetFigFont#1#2#3#4#5{%
  \reset@font\fontsize{#1}{#2pt}%
  \fontfamily{#3}\fontseries{#4}\fontshape{#5}%
  \selectfont}%
\fi\endgroup%
\begin{picture}(5581,3448)(1676,-3880)
\put(6601,-1561){\makebox(0,0)[lb]{\smash{{\SetFigFont{12}{14.4}{\familydefault}{\mddefault}{\updefault}{\color[rgb]{0,0,0}$\log \frac{1}{|z-w|}$}%
}}}}
\put(6751,-2761){\makebox(0,0)[lb]{\smash{{\SetFigFont{12}{14.4}{\familydefault}{\mddefault}{\updefault}{\color[rgb]{0,0,0}$l_r(z-\tilde z)+k$}%
}}}}
\put(6526,-3811){\makebox(0,0)[lb]{\smash{{\SetFigFont{12}{14.4}{\familydefault}{\mddefault}{\updefault}{\color[rgb]{0,0,0}$w$}%
}}}}
\put(5701,-3511){\makebox(0,0)[rb]{\smash{{\SetFigFont{12}{14.4}{\familydefault}{\mddefault}{\updefault}{\color[rgb]{0,0,0}$z_0$}%
}}}}
\put(6001,-3511){\makebox(0,0)[lb]{\smash{{\SetFigFont{12}{14.4}{\familydefault}{\mddefault}{\updefault}{\color[rgb]{0,0,0}$\tilde z$}%
}}}}
\put(3301,-3586){\makebox(0,0)[lb]{\smash{{\SetFigFont{12}{14.4}{\familydefault}{\mddefault}{\updefault}{\color[rgb]{0,0,0}$S_V=\overline{\D}$}%
}}}}
\end{picture}%

%% file: linebd.pspdftex
\begin{picture}(0,0)%
\includegraphics{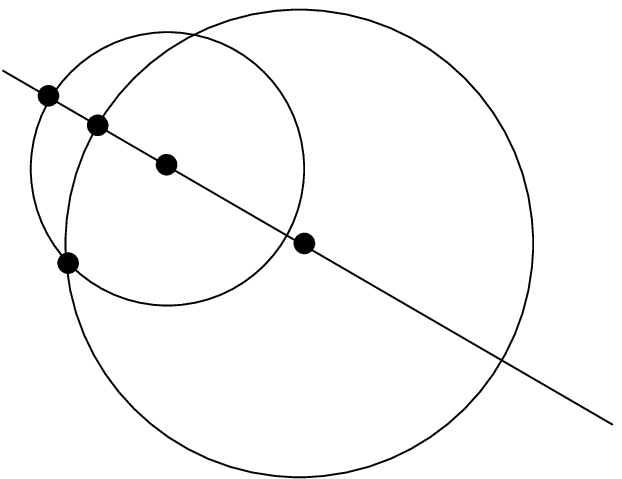}%
\end{picture}%
\setlength{\unitlength}{4144sp}%
\begingroup\makeatletter\ifx\SetFigFont\undefined%
\gdef\SetFigFont#1#2#3#4#5{%
  \reset@font\fontsize{#1}{#2pt}%
  \fontfamily{#3}\fontseries{#4}\fontshape{#5}%
  \selectfont}%
\fi\endgroup%
\begin{picture}(2814,2154)(769,-2398)
\put(3061,-691){\makebox(0,0)[lb]{\smash{{\SetFigFont{12}{14.4}{\familydefault}{\mddefault}{\updefault}{\color[rgb]{0,0,0}$C_2$}%
}}}}
\put(1621,-916){\makebox(0,0)[lb]{\smash{{\SetFigFont{12}{14.4}{\familydefault}{\mddefault}{\updefault}{\color[rgb]{0,0,0}$\tilde z$}%
}}}}
\put(2296,-1366){\makebox(0,0)[lb]{\smash{{\SetFigFont{12}{14.4}{\familydefault}{\mddefault}{\updefault}{\color[rgb]{0,0,0}$w$}%
}}}}
\put(1261,-691){\makebox(0,0)[rb]{\smash{{\SetFigFont{12}{14.4}{\familydefault}{\mddefault}{\updefault}{\color[rgb]{0,0,0}$b$}%
}}}}
\put(991,-556){\makebox(0,0)[rb]{\smash{{\SetFigFont{12}{14.4}{\familydefault}{\mddefault}{\updefault}{\color[rgb]{0,0,0}$a$}%
}}}}
\put(991,-1456){\makebox(0,0)[rb]{\smash{{\SetFigFont{12}{14.4}{\familydefault}{\mddefault}{\updefault}{\color[rgb]{0,0,0}$z$}%
}}}}
\put(2071,-601){\makebox(0,0)[lb]{\smash{{\SetFigFont{12}{14.4}{\familydefault}{\mddefault}{\updefault}{\color[rgb]{0,0,0}$C_1$}%
}}}}
\put(3421,-2041){\makebox(0,0)[lb]{\smash{{\SetFigFont{12}{14.4}{\familydefault}{\mddefault}{\updefault}{\color[rgb]{0,0,0}$\cal L$}%
}}}}
\end{picture}%

%% file: nbd.pspdftex
\begin{picture}(0,0)%
\includegraphics{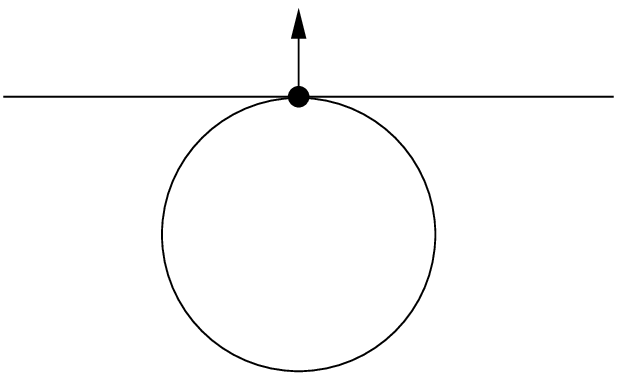}%
\end{picture}%
\setlength{\unitlength}{4144sp}%
\begingroup\makeatletter\ifx\SetFigFont\undefined%
\gdef\SetFigFont#1#2#3#4#5{%
  \reset@font\fontsize{#1}{#2pt}%
  \fontfamily{#3}\fontseries{#4}\fontshape{#5}%
  \selectfont}%
\fi\endgroup%
\begin{picture}(2817,1680)(346,-1774)
\put(361,-331){\makebox(0,0)[lb]{\smash{{\SetFigFont{12}{14.4}{\familydefault}{\mddefault}{\updefault}{\color[rgb]{0,0,0}$H^-$}%
}}}}
\put(361,-826){\makebox(0,0)[lb]{\smash{{\SetFigFont{12}{14.4}{\familydefault}{\mddefault}{\updefault}{\color[rgb]{0,0,0}$H^+$}%
}}}}
\put(1891,-1546){\makebox(0,0)[lb]{\smash{{\SetFigFont{12}{14.4}{\familydefault}{\mddefault}{\updefault}{\color[rgb]{0,0,0}$\D$}%
}}}}
\put(1711,-691){\makebox(0,0)[b]{\smash{{\SetFigFont{12}{14.4}{\familydefault}{\mddefault}{\updefault}{\color[rgb]{0,0,0}$z$}%
}}}}
\put(1756,-376){\makebox(0,0)[lb]{\smash{{\SetFigFont{12}{14.4}{\familydefault}{\mddefault}{\updefault}{\color[rgb]{0,0,0}$\bar n$}%
}}}}
\end{picture}%

%% file: B0.pspdftex
\begin{picture}(0,0)%
\includegraphics{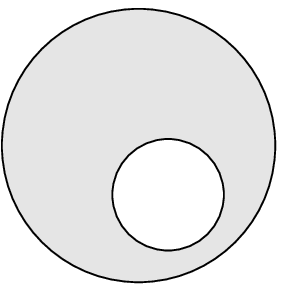}%
\end{picture}%
\setlength{\unitlength}{4144sp}%
\begingroup\makeatletter\ifx\SetFigFont\undefined%
\gdef\SetFigFont#1#2#3#4#5{%
  \reset@font\fontsize{#1}{#2pt}%
  \fontfamily{#3}\fontseries{#4}\fontshape{#5}%
  \selectfont}%
\fi\endgroup%
\begin{picture}(1266,1266)(1078,-1774)
\put(1801,-1501){\makebox(0,0)[lb]{\smash{{\SetFigFont{12}{14.4}{\familydefault}{\mddefault}{\updefault}{\color[rgb]{0,0,0}$B$}%
}}}}
\end{picture}%